\documentclass[11pt]{article}
\usepackage{fullpage, graphicx, amsmath, amsthm, amsfonts, hyperref,subfigure,tikz}
\def\uhat{\widehat u}
\newcommand\myprime{\mkern-0.5mu\raise0.4ex\hbox{$\scriptstyle\prime$}}
\input amssym.def

\newtheorem{theorem}{Theorem}
\newtheorem{lemma}[theorem]{Lemma}

\newtheorem{corollary}[theorem]{Corollary}

\newtheorem{claim}[theorem]{Claim}
\newtheorem{remark}[theorem]{Remark}
\newcounter{property}
\newtheorem{property}[property]{Property}

\newcommand{\ignore}[1]{}
\DeclareMathOperator\arctanh{arctanh}
\def\ltwob{{\lambda}}
\newcommand{\be}{\begin{equation}}
\newcommand{\ee}{\end{equation}}
\newcommand{\bea}{\begin{eqnarray}}
\newcommand{\eea}{\end{eqnarray}}
\newcommand{\bean}{\begin{eqnarray*}}
\newcommand{\eean}{\end{eqnarray*}}

\usepackage{enumitem}
\makeatletter
\newcommand{\mylabel}[2]{#2\def\@currentlabel{#2}\label{#1}}
\makeatother

\usepackage{color}

\newcommand\logderiv{{\mathcal D}}

\begin{document}

\title{Diversity in Evolutionary Dynamics}

\author{Yuval Rabani\thanks{The Rachel and Selim Benin School of Computer Science and Engineering, 
		The Hebrew University of Jerusalem, Jerusalem, Israel. yrabani@cs.huji.ac.il.
		Supported in part by ISF grants 3565-21 and 389-22, and by BSF grant 2023607.} \and
             Leonard J. Schulman\thanks{Division of Engineering and Applied Science, California Institute of Technology, Pasadena CA 91125, USA.
schulman@caltech.edu.  Supported in part by NSF grant CCF-2321079.} \and
             Alistair Sinclair\thanks{Computer Science Division, University of California, Berkeley CA 94720, USA.  sinclair@cs.berkeley.edu.
             Supported in part by NSF grant CCF-2231095.}}
             
\date{\today}

\maketitle

\begin{abstract}
We consider the dynamics imposed by natural selection on the populations
of two competing, sexually reproducing, haploid species.
In this setting, the fitness of any genome varies over time due to the changing population
mix of the competing species; crucially, this fitness variation arises naturally from the
model itself, without the need for imposing it exogenously as is typically the case.
Previous work on this model~\cite{PilSch} showed that, in the special case where each of the two
species exhibits just two phenotypes, genetic diversity is maintained at all times.
This finding supported the tenet that sexual reproduction is advantageous because it promotes
diversity, which increases the survivability of a species. 

In the present paper we consider the more realistic case where there are more than two phenotypes
available to each species. The conclusions about diversity in general turn out to be very different
from the two-phenotype case.

Our first result is negative: namely, we show that sexual reproduction
does not guarantee the maintenance of diversity at all times, i.e.,
the result of~\cite{PilSch} does \emph{not} generalize. Our counterexample
consists of two competing species with just three phenotypes each. We show that,
for any time~$t_0$ and any $\varepsilon>0$, there is a time $t\ge t_0$ at which
the combined diversity of both species is smaller than~$\varepsilon$.
Our main result is a complementary positive statement, which says that
in any non-degenerate example, diversity \emph{is} maintained in a weaker, 
``infinitely often'' sense.

Thus, our results refute the supposition that sexual reproduction ensures 
diversity at all times, but affirm a weaker assertion that extended periods of high diversity
are necessarily a recurrent event. 
\end{abstract}

\thispagestyle{empty}
\newpage
\setcounter{page}{1}

\section{Introduction}

One of the most standard models in evolutionary biology is
\emph{replicator dynamics} (see, e.g., \cite{TJ,LosertAkin83,HofSig,Sandholm}),
a non-linear dynamical system that describes the evolution of a 
species under selection and (haploid) sexual
reproduction. Individuals (genotypes) are represented as
strings of $n$ genes, each of which may take on a finite number (for simplicity 
we assume two, denoted 0 and~1) of values (alleles).  
The state of the population at time~$t$ is
described by the genotype {\it frequencies\/} $\{p_x(t)\}$, a probability distribution 
over $x\in\{0,1\}^n$.  Selection is determined by the set of real-valued {\it fitnesses\/}~$\{u(x)\}$,
which govern the likelihood of survival for each genotype~$x$ as will be specified shortly.

We focus on the important regime of {\it weak selection}, which means that reproduction
takes place at a much faster rate than selection.
In this regime, a classical result of 
Nagylaki~\cite{Nagylaki} shows that, after a short time, the dynamical system
remains very close to the so-called ``Wright manifold" of {\it product distributions\/} on
$\{0,1\}^n$; hence it suffices to keep track of only the marginal probabilities
$\{p_{i,b}(t)\}$, where for each $i\in[n]$ and $b\in\{0,1\}$, 
$p_{i,b}(t) = \sum_{x:x_i=b}p_x(t)$ is the probability that the $i$th gene
has allele~$b$.  The frequency of a given genotype~$x$ will then just be the product
$p_x(t) = \prod_{i\in[n]} p_{i,x_i}(t)$.

The replicator dynamics (in continuous time) for the marginals is specified by the system of
differential equations
\begin{equation}\label{eq:replicator1}
   \dot p_{i,b} = p_{i,b}\bigl(\uhat_{i,b}(p) - \uhat(p)\bigr),
\end{equation}
where $\uhat(p) = \sum_x p_x u(x)$ is the average fitness of the current population~$p$,
and $\uhat_{i,b}(p)$ is this same average conditioned on $x_i=b$.  Thus the quantity in
parentheses in~\eqref{eq:replicator1} is the differential fitness of having allele~$b$ in the
$i$th gene.

Chastain {\it et al.}~\cite{CLPV} 
provided an interesting
game-theoretic perspective on the replicator dynamics: namely, if one views the genes
as $n$ players in a coordination game, their allele values as actions, the genome frequencies
as mixed strategies, and common payoffs determined by the fitnesses, then \eqref{eq:replicator1}
can be viewed as the classical multiplicative weight updates
algorithm (MWU) (in continuous time), executed on the corresponding coordination game. 
However, a serious drawback of modeling evolution through these dynamics,
as was proved (for discrete time) by Mehta {\it et al.}~\cite{MPP},
is that, under minimal genericity assumptions, and except for a set of initial conditions of measure
zero, the dynamics always converges to a {\it pure\/} Nash equilibrium, i.e., a genetic monoculture.  
(This is in spite of the potential presence of uncountably many mixed Nash equilibria.)  This violates
a central tenet of evolutionary biology, which is that populations in nature
maintain diversity (see, e.g., \cite{Bartonetal}).  
In order to capture this fundamental property, the model must be enriched.

The main obstacle to diversity is the fixed set of fitnesses $\{u(x)\}$,
so one might postulate some mechanism whereby $u(x)$ varies over
time, e.g., under a changing external environment.  This approach  is taken, e.g., by
\cite{Vijayetal, MPPTV,Weitzetal,WVA}, who introduce an exogenous rule for varying fitnesses.
However, an arguably more natural approach, which does not require the
introduction of any additional assumptions, is
to incorporate the observation that fitnesses are in fact determined by 
interactions with other species, which are themselves undergoing analogous evolutionary
processes. This viewpoint, as a general evolutionary principle,
was apparently first articulated by Ehrlich and Raven~\cite{ER64}, and played a role
in the development of the so-called
``Red Queen hypothesis"~\cite{Bell,Hamilton,Hartung,Jaenike,VanValen}\footnote{The
term ``Red Queen" was first coined by Van Valen~\cite{VanValen} in his work on the
extinction of populations; the terminology refers to the Red Queen in Lewis Carroll's
{\it Alice Through the Looking Glass}, who says of Looking Glass Land: 
``Now, {\it here}, you see, it takes all the running you can do, to keep in the same place."
Later researchers, including Bell~\cite{Bell},
Hamilton~\cite{Hamilton}, Hartung~\cite{Hartung} and Jaenike~\cite{Jaenike}, advanced
the same principle as an explanation for the prevalence of sexual reproduction in nature.}
which holds that the reason for sexual reproduction is 
precisely that competition among species favors 
mixed strategies, and that sexual reproduction is a way of ensuring that the population as a whole plays a mixed strategy.

This leads to the following generalization of the replicator dynamics proposed
in the evolutionary context by~\cite{PilSch}. For simplicity we 
consider just two species, with $n$ and $m$ genes respectively.  Game-theoretically, we think
of two teams, $A$ and~$B$, with $|A|=n$ and $|B|=m$.  The players within each team play
a coordination game as before, but the two teams between them play a competitive, zero-sum game. 
Thus for each choice of actions $x\in\{0,1\}^n$, $y\in\{0,1\}^m$ by the
$m+n$ players, all players in team~$A$ receive common payoff $u(x,y)$ and all players in
team~$B$ receive $-u(x,y)$. In contrast to~\eqref{eq:replicator1}, $u$ is now a $2^n \times 2^m$ matrix. 

Writing $p=(p_1,\ldots,p_n)$ and $q=(q_1,\ldots,q_m)$ for the (time dependent) mixed strategies of teams~$A,B$ 
respectively, the replicator dynamics under weak selection becomes
\begin{equation}\label{eq:replicator2}
\begin{array}{r@{}l}
   \dot p_{i,b} &{}= p_{i,b}\bigl(\uhat_{i,b}(p,q) - \uhat(p,q)\bigr) \\[5pt]
   \dot q_{j,b} &{}= q_{j,b}\bigl(-\uhat_{j,b}(p,q) + \uhat(p,q)\bigr)
\end{array}
\end{equation}
where $\uhat(p,q)=\sum_{x,y}p_xq_yu(x,y)$ is the average payoff to team~$A$ 
under the current strategies $p,q$, $\uhat_{i,b}(p,q)$ is this same average
conditioned on player $i$ of team~$A$ choosing allele~$b$,
and $\uhat_{j,b}(p,q)$ is this same average
conditioned on player $j$ of team~$B$ choosing allele~$b$. (We consistently distinguish between the teams by the $i,j$ subscripts.)
 In light of the discussion in the previous paragraph,
we call~(\ref{eq:replicator2}) the \emph{Red Queen Dynamics}.

Note that, superficially, the resulting $(m+n)$-player game resembles a two-player zero-sum
game between players~$A$ and~$B$. The crucial distinction is that here
teams $A$ and~$B$ are restricted to mixed strategies that are \emph{product distributions} over
their individual gene choices (i.e., the individual members of each team play independently,
rather than in coordinated fashion as in the two-player game).  
Thus, for example, the two-player
game may possess mixed equilibria that are not even available in the $(m+n)$-player game,
while an equilibrium in the latter may not be an equilibrium in the former because
of the larger class of allowed defections in the two-player game.  
Indeed, it is not hard to find
examples of two-team zero-sum games for which the classical von Neumann minimax theorem
for two-person games fails, and the team that picks its strategy second has an advantage~\cite{vSK}.
This phenomenon was  investigated systematically
by Schulman and Vazirani~\cite{SchulVaz}, 
who gave tight bounds on the so-called {\it duality gap}, which quantifies the cost to the teams of having to play statistically-independent (i.e., product) strategies vs.\ being able to use a coordinated strategy.

Since we have simplified to the case of only two alleles per gene, we can abbreviate by setting $p_i=p_{i,0},q_j=q_{j,0}$, and the dynamics~\eqref{eq:replicator2} can be rewritten
\begin{equation}\label{eq:replicator3}
\begin{array}{r@{}l}
   \dot p_{i} &{}= p_{i}\bigl(\uhat_{i,0}(p,q) - \uhat(p,q)\bigr)  \\[5pt]
  \dot q_{j} &{}= q_{j}\bigl(-\uhat_{j,0}(p,q) + \uhat(p,q)\bigr)
\end{array}
\end{equation}

Recall that our primary motivation for considering the Red Queen dynamics was to exhibit a
plausible model that guarantees the maintenance of diversity over time. We measure
diversity in terms of the \emph{entropy} of the species populations given by $p(t)=(p_1(t),\ldots,p_n(t)),q(t)=(q_1(t),\ldots,q_m(t))$. 
We write $(p,q)\in (0,1)^{n+m}$ to indicate that every $p_i$ and $q_j$ are in the open interval $(0,1)$; i.e., all genotypes are represented in the population.

Ideally one would
like to claim that the following diversity property holds:
\begin{property} For any initial populations $(p(0),q(0))\in (0,1)^{n+m}$ (i.e., all genotypes are initially represented), 
\bea \label{eq:entcond1} \exists
\varepsilon>0, t_0\geq 0 \text{ s.t.\ } \forall t\geq t_0, \quad
H(p(t))+H(q(t)) \ge \varepsilon,
\eea
where $H(p)=\sum_i H(p_i)$ and $H(p_i)=- p_i\log p_i-(1-p_i)\log(1-p_i)$ is Shannon entropy
(and $H(q)$ is defined analogously).  In other words, at all times
the two-species system maintains diversity (not necessarily within a particular species).
\label{propPapaBear} \end{property}

This property is known to hold in some very restricted cases: in single-gene species
(i.e., $n=m=1$~\cite{PilSham}) and, for general $n,m$, in
``Boolean phenotype" games~\cite{PilSch}, those games in which
the payoff $u(x,y)$ is allowed to depend only on Boolean functions of~$x$ and~$y$.

This paper is devoted to exploring in some detail the extent to which Property~\ref{propPapaBear},
or similar diversity statements, hold in more general dynamics, based on the nature of the game-theoretic
equilibria of the two-team game. 

If the two-team game has a strict pure Nash equilibrium, then (as we will show later in Theorem~\ref{thm:main}\ref{2iii})
the dynamics initialized in a neighborhood of that point will converge to it.  Therefore, the interesting case is when
strict pure Nash equilibria do not exist.  Our first result shows that even in this case---and indeed,
even in the absence of {\it weak\/} pure equilibria---there are examples for which Property~\ref{propPapaBear}
fails.  Indeed, in these examples the maintenance of entropy expressed in Equation~\eqref{eq:entcond1} fails to hold for all but
a measure-zero set of  initial populations.

\begin{theorem}\label{thm:strongfail}
There exist Red Queen dynamics with no weak pure Nash equilibrium for
which Property~\ref{propPapaBear} fails.   Even stronger, inequality~\eqref{eq:entcond1} holds
only on a set of initial conditions $(p(0),q(0))\in (0,1)^{n+m}$ of measure zero.
\end{theorem}
Our concrete example used in the proof of Theorem~\ref{thm:strongfail} consists of species with
$n=m=2$, i.e., each having four genotypes.  However, the payoff function causes two of these
genotypes in each species to behave identically, so the number of phenotypes in each case is only three.
Interestingly, this is in sharp contrast to the result of~\cite{PilSch} mentioned above, which shows that
with only two phenotypes per species diversity is always maintained.

Theorem~\ref{thm:strongfail} raises the possibility that diversity might vanish even in a competitive co-evolution (Red Queen) dynamics. 
However, we show that this is false.  Complementing Theorem~\ref{thm:strongfail}, in our main result we show that, under modest game-theoretic
conditions, an entropy lower bound statement (weaker, necessarily, than Property~\ref{propPapaBear}) does in fact hold. This weaker statement is the
following:
\begin{property}
For any initial populations $(p(0),q(0))\in (0,1)^{n+m}$, 
\begin{equation}\label{eq:entcond2} \exists \varepsilon>0 \text{ s.t.\ }
   \int_0^\infty \max\{0, H(p(t)) + H(q(t)) - \varepsilon\} \ dt = \infty.
\end{equation}
In other words, during an infinite span of time the entropy is uniformly bounded away from~$0$.
\label{propMamaBear} \end{property}

Theorem~\ref{thm:main}\ref{2i}, below, shows
that Property~\ref{propMamaBear} holds for all Red Queen dynamics, 
subject only to the condition that no pure Nash equilibrium exists.
Moreover, in Theorem~\ref{thm:main}\ref{2ii} we show that, under the even weaker condition that
no \emph{strict} pure Nash equilibrium exists, 
a slightly weaker diversity property is guaranteed to hold, namely:
\begin{property} 
For any initial populations $(p(0),q(0))\in (0,1)^{n+m}$, 
\be \int_0^\infty (H(p(t)) + H(q(t))) \ dt = \infty. \label{eq:entcond3} \ee
\label{propBabyBear} \end{property} 
Finally, we complete the picture by proving in Theorem~\ref{thm:main}\ref{2iii} the complementary result 
that, in the presence of a strict pure equilibrium,
even this weaker property fails on a set of initial conditions of positive measure. 

These general results are gathered in the following theorem.

\begin{theorem}\label{thm:main}
The following results hold for general Red Queen dynamics.
\begin{enumerate}[label=\textnormal{(\roman*)}] 
\item[\mylabel{2i}{(i)}]  If the dynamics has no pure Nash equilibrium, 
then Property~\ref{propMamaBear} holds.
\item[\mylabel{2ii}{(ii)}]  
If the assumption in (i) is weakened to assume only that the 
dynamics has no \emph{strict} pure Nash equilibrium,
then Property~\ref{propBabyBear} holds.
\item[\mylabel{2iii}{(iii)}]  
If the dynamics has a strict pure Nash equilibrium, then Property~\ref{propBabyBear} (and therefore also Property~\ref{propMamaBear}) 
fails on a set of initial populations $(p(0),q(0))\in(0,1)^{n+m}$ of positive measure.
\end{enumerate}
\end{theorem}
 
\par\medskip\noindent
{\bf Interpretation of Results:} Theorem~\ref{thm:strongfail} refutes the supposition that sexual reproduction (as expressed in
the weak selection model) ensures diversity of populations at all times, by exhibiting a concrete
example in which entropy becomes arbitrarily small infinitely often.  On the other hand, Theorem~\ref{thm:main}
affirms the weaker assertion that, even in games satisfying the hypotheses of Theorem~\ref{thm:strongfail},
extended periods of high diversity recur over an unbounded sequence of time intervals.
In particular, our analysis in Theorems~\ref{thm:strongfail} and~\ref{thm:main} demonstrates
the phenomenon that two competing populations may experience recurring periods during which they simultaneously
approach genetic monocultures, despite the fact that one limit genome is poorly fit against the competitor's 
limit genome. However, when this happens the population dynamics eventually causes the diversity of the poorly
fit species to increase for a significant amount of time, during which it escapes the poorly fit genome.
\par\medskip\noindent
{\bf Organization of the Paper:} The remainder of the paper is organized as follows.  We begin in Section~\ref{sec:prelims} by deriving
several alternative formulations of the Red Queen dynamics that will be useful in our subsequent analysis.
Section~\ref{sec:thm1proof} is devoted to a proof of Theorem~\ref{thm:strongfail}; it proceeds by
constructing an explicit example dynamics for which entropy is not maintained over time.  In Section~\ref{sec:thm2proof}
we prove our main result, showing that a weaker form of entropy maintenance \emph{does} hold for general Red 
Queen dynamics under mild assumptions about the existence of Nash equilibria.  We conclude with a discussion
of our results and some open problems in Section~\ref{sec:discussion}.

\section{Preliminaries}\label{sec:prelims}
First we reformulate the dynamics from the introduction.
For $(x,y)\in\{0,1\}^{n+m}$, $(p,q)\in (0,1)^{n+m}$, define
\begin{equation}\label{eqn:KAdef}
K_{A,i}(x,y,p,q) := \prod_{i'\ne i} p_{i',x_{i'}} \prod_j q_{j,y_j}\,;
\end{equation}
this is the probability under the product distribution $(p,q)$ of the pair of genomes 
$(x_{-i},y)$, where $x_{-i}$ is the genome $x$~excluding $x_i$.  Analogously, define 
\begin{equation}\label{eqn:KBdef}
K_{B,j}(x,y,p,q) := \prod_{i} p_{i,x_i} \prod_{j'\ne j} q_{j',x_{j'}}\,. 
\end{equation} 
We may now rewrite the replicator dynamics~\eqref{eq:replicator2} as follows: 
\begin{eqnarray}
\dot p_{i,b} & = & p_{i,b}\bigl(\uhat_{i,b}(p,q) - \uhat(p,q)\bigr) \nonumber \\
& = & p_{i,b}\sum_{\substack{x\in\{0,1\}^n \\ x_i=b}} \sum_{y\in\{0,1\}^m} \left(K_{A,i}(x,y,p,q)\cdot u(x,y) - 
   K_{A,i}(x,y,p,q) \bigl(p_{i,b}\cdot u(x,y) + p_{i,\bar{b}}\cdot u(\bar x^i,y)\bigr)\right)\nonumber \\
& = & p_{i,b}(1 -  p_{i,b})\cdot \sum_{\substack{x\in\{0,1\}^n \\ x_i=b}} \sum_{y\in\{0,1\}^m} K_{A,i}(x,y,p,q) \bigl( u(x,y) - u(\bar x^i,y)\bigr) \label{eq:replicator4pprelims}\\
& = & p_{i,b}(1 -  p_{i,b}) \bigl( \uhat_{i,b}(p,q) - \uhat_{i,\bar b}(p,q)\bigr), \label{eq:replicator4p2prelims}
\end{eqnarray}
where $\bar{b}=1-b$ denotes the complement of $b\in\{0,1\}$, and $\bar x^i$ denotes the vector $x\in\{0,1\}^n$ with
its $i$th coordinate replaced by its complement.

Similarly, we have 
\begin{eqnarray}\label{eq:replicator4q}
\dot q_{j,b} &=& -q_{j,b}(1 -  q_{j,b})\cdot \sum_{x\in\{0,1\}^n}\sum_{\substack{y\in\{0,1\}^m \\ y_j=b}}  K_{B,j}(x,y,p,q) \bigl(u(x,y)-u(x,\bar y^j)\bigr) \label{eq:replicator4qprelims} \\
                    &=& -q_{j,b}(1 -  q_{j,b}) \bigl(\uhat_{j,b}(p,q)-\uhat_{j,\bar b}(p,q)\bigr).  \label{eq:replicator4q2prelims} 
\end{eqnarray}
We note that, since the differences in utilities on the right-hand sides of~\eqref{eq:replicator4p2prelims} and~\eqref{eq:replicator4q2prelims}
are uniformly bounded, any trajectory initiated at an interior point $(p(0),q(0))\in (0,1)^{n+m}$ remains in the interior (i.e., no
value $p_{i,b}(t)$ or $q_{j,b}(t)$ becomes zero) at all finite times~$t$.

Following the simplified notation in~\eqref{eq:replicator3}, where we replace $p_{i,0}$ by~$p_i$ and $q_{j,0}$ by~$q_j$,
we get the following simplified form of the dynamics in Equations~\eqref{eq:replicator4p2prelims} and~\eqref{eq:replicator4q2prelims}:
\begin{equation}\label{eq:replicator4}
\begin{array}{r@{}l}
   \dot p_{i} &{} = p_{i}(1-p_i) \bigl(\uhat_{i,0}(p,q) - \uhat_{i,1}(p,q)\bigr); \\[5pt] 
   \dot q_{j} &{} = -q_{j}(1-q_j) \bigl(\uhat_{j,0}(p,q) - \uhat_{j,1}(p,q)\bigr).
\end{array}
\end{equation}

For our proof of Theorem~\ref{thm:strongfail}, it will be convenient to rewrite these dynamics in terms of $\delta\in [-1,+1]^n$
and $\epsilon\in [-1,+1]^m$, linear transforms of $p$ and $q$, defined by
$\delta_i= 2p_{i} - 1$, $\epsilon_j = 2q_{j} - 1$. Transform the
payoff matrix $u$ to the Fourier basis as follows. For $r\in\{0,1\}^n$ and $s\in\{0,1\}^m$
write
$$
C_{r,s} = 2^{-n-m} \sum_{x\in\{0,1\}^n}\sum_{y\in\{0,1\}^m} (-1)^{r\cdot x+ s\cdot y} u(x,y).
$$
Observe that 
\begin{align*} K_{A,i}(x,y,p,q) &= 2^{-m-n+1} \Bigl(\prod_{i'\neq i} (1+\delta_i)^{1-x_i}(1-\delta_i)^{x_i}\Bigr)\Bigl(\prod_j (1+\epsilon_j)^{1-y_j}(1-\epsilon_j)^{y_j}\Bigr) \\
&= 2^{-m-n+1} \Bigl(\prod_{i'\neq i} (1+(-1)^{x_i}\delta_i)\Bigr)\Bigl(\prod_j (1+(-1)^{y_j}\epsilon_j)\Bigr).
\end{align*}
Expanding~\eqref{eq:replicator4}, we get
\begin{align*} \dot \delta_i &= 2 \dot p_i = \frac12 (1-\delta_i^2) 
\sum_{\substack{x\in\{0,1\}^n \\ x_i=0}} \sum_{y\in\{0,1\}^m} \bigl(u(x,y)-u(\bar x^i,y)\bigr) K_{A,i}(x,y,p,q) \\
&= 2^{-m-n} (1-\delta_i^2) \sum_{\substack{x\in\{0,1\}^n \\ x_i=0}} \sum_{y\in\{0,1\}^m} \bigl(u(x,y)-u(\bar x^i,y)\bigr)
\Bigl(\prod_{i'\neq i} (1+(-1)^{x_{i'}}\delta_{i'})\Bigr)\Bigl(\prod_j (1+(-1)^{y_j}\epsilon_j)\Bigr) \end{align*}
Letting $\delta^r=\prod_{k=1}^n \delta_k^{r_k}$ and $\epsilon^s=\prod_{k=1}^m \epsilon_k^{s_k}$, this becomes
\begin{align} \dot \delta_i 
&= 2^{-m-n} (1-\delta_i^2) \sum_{\substack{x\in\{0,1\}^n \\ x_i=0}} \sum_{y\in\{0,1\}^m} \bigl(u(x,y)-u(\bar x^i,y)\bigr) 
\sum_{\substack{r\in\{0,1\}^n \\ r_i=0}} \sum_{s\in\{0,1\}^m} (-1)^{r \cdot x+s \cdot y} \, \delta^r \epsilon^s \nonumber \\
 &= 2^{-m-n} (1-\delta_i^2) 
\sum_{\substack{r\in\{0,1\}^n \\ r_i=0}} \sum_{s\in\{0,1\}^m} 
\delta^r \epsilon^s
\sum_{\substack{x\in\{0,1\}^n \\ x_i=0}} \sum_{y\in\{0,1\}^m} \bigl(u(x,y)-u(\bar x^i,y)\bigr) 
(-1)^{r\cdot x + s \cdot y} \nonumber \\
&=(1-\delta_i^2) \sum_{\substack{r\in\{0,1\}^n \\ r_i=0}} \sum_{s\in\{0,1\}^m} 
C_{\bar r^i,s}\, \delta^r \epsilon^s. \label{eq: delta dot}
\end{align}
By an analogous calculation,
\begin{align}
 \dot{\epsilon}_j
& = -(1-\epsilon_j^2)
\sum_{r\in\{0,1\}^n}  \sum_{\substack{s\in\{0,1\}^m \\ s_j=0}} 
         C_{r,\bar s^j} \delta^r \epsilon^s. \label{eq: epsilon dot}
\end{align}

Finally, there is a third interesting representation for the dynamics. For $i\in[n]$ let $\alpha_i=\arctanh \delta_i$ and $\beta_j=\arctanh \epsilon_j$. Then~\eqref{eq: delta dot},~\eqref{eq: epsilon dot} become
\begin{align}
\dot \alpha_i &=  \sum_{\substack{r\in\{0,1\}^n \\ r_i = 0}} 
\sum_{s\in\{0,1\}^m}
C_{\bar r^i,s} \Bigl(\prod_{i'\ne i} (\tanh \alpha_{i'})^{r_{i'}}\Bigr)
\Bigl(\prod_j (\tanh \beta_j)^{s_j}\Bigr)
\label{eq: alpha dot}\\
\dot \beta_j & = - \sum_{r\in\{0,1\}^n}  \sum_{\substack{s\in\{0,1\} \\ s_j = 0}}
         C_{r,\bar s^j} \Bigl(\prod_i (\tanh \alpha_i)^{r_i}\Bigr)
\Bigl(\prod_{j'\ne j} (\tanh \beta_{j'})^{s_{j'}}\Bigr) \label{eq: beta dot}
\end{align}
We will not need this representation in our proofs, but it has the advantage that the $(1-\delta^2)$, $(1-\epsilon^2)$
terms that damp progress near the boundaries of the cube are not present in this representation. 
Because of those damping terms, it is difficult in simulations to see any progress at all (in the $(\delta,\epsilon)$ representation)
when the trajectory is near a corner; and, depending on the details of the game, the trajectory might spend most of its time near corners. 
By contrast, in the $(\alpha,\beta)$ representation the corner regions are effectively expanded, and the trajectory
there is illuminated in simulations; see Fig.~\ref{fig:simulation}.

\section{Proof of Theorem~\ref{thm:strongfail}} \label{sec:thm1proof}

\subsection{The exemplifying game}

We give a specific example of Red Queen dynamics with the claimed behavior.  
The dynamics is defined by the following payoff matrix $u$, with $n=m=2$:
\begin{equation}\label{nicequadgame}
\left[\begin{array}{rrrr} 
 1 & z & z & -1 \\ 
-z & 0 & 0 & -z \\ 
-z & 0 & 0 & -z \\ 
-1 & z & z &  1 
\end{array}\right].
\end{equation}
Here the rows are indexed by strategies for Team~A, and the columns by strategies for Team~B.
Row and column indices, reading top to bottom and left to right respectively, are $00$, $01$, $10$, $11$.
Note that, since the second and third rows (respectively, columns) of the matrix are identical, each of
the two species~A and~B has four genotypes but only three phenotypes.

For definiteness we set $z = \frac 1 2$ (but the analysis below in fact holds for any $z\in (0,1)$).
It is easy to verify that in this case the representation of~$u$ in the Fourier basis is given
by the Fourier coefficients:
\begin{align*}
& C_{00,00} = C_{00,01} = C_{00,10} = C_{01,00} = C_{10,00} = C_{11,10} = C_{11,01} = C_{10,11} = C_{01,11} = C_{11,11} = 0,\\
& C_{01,01} = C_{01,10} = C_{10,01} = C_{10,10} = C_{11,00} = -C_{00,11} = \frac 1 4.
\end{align*}
Thus, plugging these values into Equations~\eqref{eq: delta dot} and~\eqref{eq: epsilon dot}, we
get the following dynamics in the variables $\epsilon_i(t),\delta_i(t)$, for $i\in\{1,2\}$:
\begin{eqnarray}
\dot{\delta}_i & = & \frac 1 4 (1 - \delta_i^2) (\delta_{3-i} + \epsilon_1 + \epsilon_2), \label{eqn:deltadot}\\
\dot{\epsilon}_i & = & -\frac 1 4 (1 - \epsilon_i^2)(\delta_1 + \delta_2 - \epsilon_{3-i}).  \label{eqn:epsdot}
\end{eqnarray}

We note first that it is easy to verify by a simple case analysis that this four-gene game possesses a 
unique Nash equilibrium, namely where all players play a mixed strategy 
(the uniform distribution on $\{0,1\}$, which corresponds to the origin $\delta_1=\delta_2=\epsilon_1=\epsilon_2=0$
in the above parameterization).\footnote{Interestingly, the von Neumann equilibrium of the 
underlying zero-sum game on genomes is different: both players play a uniform distribution on
$\{00,11\}$. Of course, this strategy does not lie on the four-player Wright manifold.}
Thus no pure Nash equilibrium exists.  Other than the origin, the only other fixed points of the dynamics
are the 16 ``corners" $\delta_i,\epsilon_i\in\{\pm 1\}$ (corresponding to pure strategies for all players);
these, however, are not Nash equilibria of our game.

It will be useful also to note that the dynamics~\eqref{eqn:deltadot}-\eqref{eqn:epsdot} commute with the 
following symmetry under the cyclic group~$C_4$:
\begin{equation*}
\delta_i' :=\epsilon_i\,;\qquad
\epsilon_i' := - \delta_i\,,
\end{equation*}
as can be verified from the form of the time derivatives:
\begin{align*} \dot{\delta_i'}&=\dot \epsilon_i = -\frac14 (1-\delta_i'^2)(-\epsilon_1'-\epsilon_2'-\delta_{3-i}') \\
\dot{\epsilon_i'} &= -\dot \delta_i = - \frac14 (1-\epsilon_i'^2)(\delta_1'+\delta_2'-\epsilon_{3-i}') \end{align*}

\begin{remark}
The above game is of a type that has been used several times in the literature. The general prescription
is that the action~$x$ of Team~A is mapped to $f(x)\in\{0,1,*\}$ by $00\to 0,\,11\to 1$, and $01,10\to *$;
a similar function $g$ is applied to the action~$y$ of Team~B; in addition, a $2 \times 2$ payoff
matrix $U$ and a scalar~$z$ are specified. The payoffs are as follows: if $(f(x),g(y))\in\{0,1\}^2$ then the payoff
is $U_{f(x),g(y)}$. If $f(x)\in \{0,1\}$ while $g(y)=*$ then the payoff is~$z$; if $f(x)=*$ while $g(y)\in \{0,1\}$ then
the payoff is $-z$; if $f(x)=g(y)=*$ then the payoff is~$0$. 

Most often in this literature $U$ has been chosen as the Matching Pennies game, i.e., $U_{0,0}=U_{1,1}=1$
and $U_{0,1}=U_{1,0}=-1$.  This was the case, e.g., in~\cite{SchulVaz},
where the game in Example~2 is exactly as in the present paper with the setting $z=1$. A similar game based on Matching Pennies was used in~\cite{PilSch}.
In~\cite{KPV-G}, which is concerned with the very different question of computing Nash equilibria, 
the game of~\cite{SchulVaz} was used but with a payoff  $z\in(0,1)$, and is referred to as ``Generalized Matching
Pennies."  This coincides with the game in the present paper with $z\in(0,1)$.
\end{remark}

\subsection{Outline of the analysis}\label{subsec:outline}

Let $\Omega=[-1,+1]^4$ denote the domain of $(\delta_1,\delta_2,\epsilon_1,\epsilon_2)$.
We define the following subsets of~$\Omega$:
\begin{eqnarray*}
S&=&\{(\delta_1,\delta_2,\epsilon_1,\epsilon_2)\in \Omega:\ [\delta_1=\delta_2]\wedge [\epsilon_1=\epsilon_2]\};\\
\check S&=&\{(\delta_1,\delta_2,\epsilon_1,\epsilon_2)\in \Omega:\ [\delta_1=-\delta_2]\wedge [\epsilon_1=-\epsilon_2] \}.
\end{eqnarray*}
We refer to $S$ and~$\check S$ as the ``symmetric" and ``skew-symmetric" subspaces respectively.  
Note from~\eqref{eqn:deltadot}-\eqref{eqn:epsdot} that both~$S$ 
and~$\check S$ are closed under the dynamics.  The symmetric subspace~$S$ corresponds to the situation
where both players in Team~A are constrained to play the same strategy, as are both players in Team~B.

At a very high level, our analysis of the game dynamics will proceed as follows.  First, we will show in Section~\ref{subsec:subspaceS}
that, for initial conditions inside the symmetric subspace~$S$ (minus the origin), property~\eqref{eq:entcond1} fails.
This is a 2-dimensional subspace with parameters $\delta=\delta_1=\delta_2$ and $\epsilon=\epsilon_1=\epsilon_2$,
and hence much easier to analyze.  We show this property
by proving that the dynamics in this subspace ``spirals outwards", getting closer and closer to each of the four
corners $\delta,\epsilon\in\{\pm1\}$ over an infinite subsequence of times.  Since the corners correspond
to monocultures, this establishes that entropy cannot be bounded away from zero.  Then, in Section~\ref{subsec:subspaceScheck},
we extend the
analysis to the whole of~$\Omega$, except for the skew-symmetric subspace~$\check S$.  This is achieved
by showing that, starting anywhere outside~$\check S$, the dynamics eventually approaches the symmetric
subspace~$S$, so that the previous analysis inside~$S$ still applies.  
Since $\check S$ has measure~0, this proves Theorem~\ref{thm:strongfail}.  (A much simpler analysis---see 
Remark~\ref{rem:skewsymm} below---shows
that trajectories inside~$\check S$ actually approach the origin, so property~\eqref{eq:entcond1} does hold there.)

\subsection{Potential functions}\label{subsec:potfuns}

It will be useful to establish some potential functions. For a real-valued function~$f$ on~$\Omega$, let 
$\logderiv f := \dot f/f$ denote the logarithmic derivative of~$f$, where $\dot f$ is the time derivative
under the dynamics.  By direct calculation from~\eqref{eqn:deltadot}-\eqref{eqn:epsdot} we find:
\begin{equation}\label{eqn:logderivs1}
\begin{array}{r@{}l}
\logderiv(1-\delta_i^2) &{} = -\frac{1}{2}(\delta_1\delta_2+\delta_i\epsilon_1+\delta_i\epsilon_2); \\[5pt]
\logderiv(1-\epsilon_i^2)&{} =-\frac{1}{2}(\epsilon_1\epsilon_2 -\delta_1\epsilon_i-\delta_2\epsilon_i),
\end{array}         
\end{equation}
and
\begin{equation}\label{eqn:logderivs2}
\begin{array}{r@{}l}
\logderiv(\delta_1-\delta_2)&{}=-\frac{1}{4}(1+\delta_1\delta_2 +(\delta_1+\delta_2)(\epsilon_1+\epsilon_2));\\[5pt]
\logderiv(\epsilon_1-\epsilon_2)&{}=-\frac{1}{4}(1+\epsilon_1\epsilon_2 -(\delta_1+\delta_2)(\epsilon_1+\epsilon_2)).
\end{array}
\end{equation}
We note in passing that, while expressions of the form~\eqref{eqn:logderivs1} hold generally as a consequence of the
dynamical equations~\eqref{eq: delta dot} and~\eqref{eq: epsilon dot}, those in~\eqref{eqn:logderivs2} for $\logderiv(\delta_i-\delta_j)$, $\logderiv(\epsilon_i-\epsilon_j)$
hold because of the specific symmetry property of the utilities in our game, namely that $u(x,y)$ is a function only of $|\{i:x_i=1\}|$ and $|\{j:y_j=1\}|$.

Since $\logderiv(fg)=\logderiv(f)+\logderiv(g)$ and $\logderiv(1/f)=-\logderiv f$, 
taking linear combinations of~\eqref{eqn:logderivs1},~\eqref{eqn:logderivs2} allows us to easily compute logarithmic derivatives of arbitrary
rational functions of the above four quantities. In particular, this allows us to construct \emph{potential functions,}
functions whose derivative is non-increasing in time.  A key role will be played by the function
\begin{equation}\label{eqn:pidef}
\pi(t) := (1-\delta_1^2(t))(1-\delta_2^2(t))(1-\epsilon_1^2(t))(1-\epsilon_2^2(t)),
\end{equation}
whose logarithmic derivative is seen from~\eqref{eqn:logderivs1} to be
\begin{equation}\label{eqn:logderivpi}
\logderiv\pi = -(\delta_1\delta_2+\epsilon_1\epsilon_2).
\end{equation}

We will also use the following pair of potential functions:
\begin{equation}\label{eqn:mudef}
  \mu_A := \frac{(\delta_1(t)-\delta_2(t))^2}{(1-\delta_1^2(t))(1-\delta_2^2(t))}; \quad
  \mu_B := \frac{(\epsilon_1(t)-\epsilon_2(t))^2}{(1-\epsilon_1^2(t))(1-\epsilon_2^2(t))},
\end{equation}
whose logarithmic derivatives are given by
\begin{equation}\label{eqn:Dmu}
   \logderiv\mu_A = -\frac{1}{2}(1-\delta_1\delta_2) \le 0; \quad
   \logderiv\mu_B = -\frac{1}{2}(1-\epsilon_1\epsilon_2) \le 0.
\end{equation}   

\begin{remark}
It may be interesting to note an additional potential function that we will not use, namely
$\sigma(t) := (\delta_1(t)-\delta_2(t))^2 (\epsilon_1(t)-\epsilon_2(t))^2$,
whose logarithmic derivative is
$\logderiv\sigma = -1-(\delta_1\delta_2+\epsilon_1\epsilon_2)/2 \leq 0$.
Observe that this, together with~\eqref{eqn:logderivpi}, implies that $\logderiv (\sigma^2/\pi) = -2$.
\end{remark}

\subsection{The dynamics in the subspace $S$}\label{subsec:subspaceS}

As indicated above, we first consider the special case in which the dynamics is confined to the closed subspace~$S$,
so that $\delta_1=\delta_2=\delta$ and $\epsilon_1=\epsilon_2=\epsilon$.  We may thus restrict to the
two-dimensional domain $\Omega':=[-1,1]^2$ with coordinates $(\delta,\epsilon)$.  In this case, we see 
from~\eqref{eqn:logderivpi} that
$\logderiv \pi = -(\delta^2+\epsilon^2) \le 0$ at all times,
and hence $\pi$ is monotonically decreasing (strictly except at the origin $\delta=\epsilon=0$).
Moreover, 
\begin{equation}\label{eqn:logpi}
\log \pi(t) = \log \pi(0) + \int_0^t \logderiv\pi(s) ds = \log \pi(0) - \int_0^t (\delta^2(s)+\epsilon^2(s)) ds.
\end{equation}

Let $\rho\in (0,1)$, and consider the open set $C_{\rho} := \{(\delta,\epsilon)\in\Omega': \delta^2+\epsilon^2-\delta^2\epsilon^2<\rho\}$.
Note that $C_\rho \supset B_\rho$ where $B_\rho \{(\delta,\epsilon)\in\Omega': \delta^2+\epsilon^2<\rho\}$ is the open $\ell_2$ ball of radius~$\rho$.  Observe also that, if the initial
state $(\delta(0),\epsilon(0))$ of the trajectory lies outside~$C_\rho$, then the time-$t$ state $(\delta(t),\epsilon(t))$
also lies outside~$C_\rho$ at all times~$t\ge 0$: this is because $\delta^2+\epsilon^2-\delta^2\epsilon^2 = 1-\sqrt\pi$, and $\pi$ is decreasing.  Since $B_\rho\subset C_\rho$, this ensures that, along any such 
trajectory, $\delta^2(t) + \epsilon^2(t)\ge\rho$ for all $t\ge 0$, 
so $\log \pi(t)\rightarrow -\infty$ as $t\rightarrow\infty$. Therefore $\pi(t)\rightarrow 0$
as $t\rightarrow\infty$, which in turn implies that 
$\max\{|\delta(t)|,|\epsilon(t)|\}\rightarrow 1$ as $t\rightarrow\infty$.

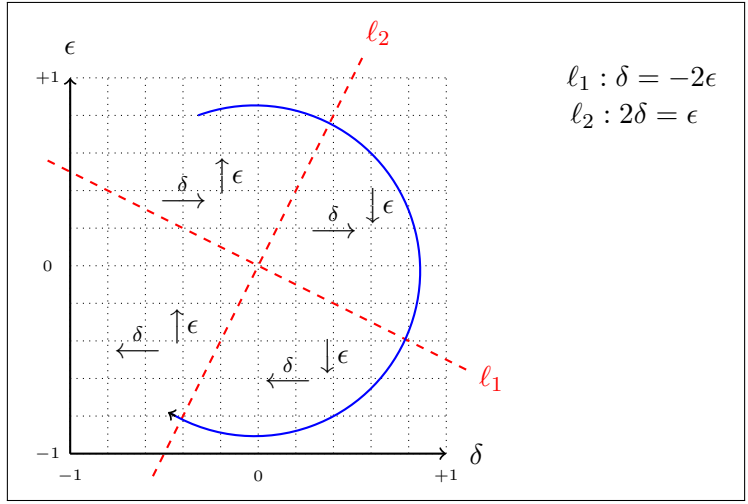
\begin{figure}
\begin{center}
\fbox{
\begin{tikzpicture}
\draw[step=0.5cm,black,dotted] (-2.5,-2.5) grid (2.5,2.5);

\draw[red,dashed,thick] (-1.4,-2.8) -- (1.4,2.8);
\draw[red] (1.6,3.1) node{$\ell_2$};
\draw (5,2.0) node{$\ell_2: 2\delta = \epsilon$};

\draw[red,dashed,thick] (-2.8,1.4) -- (2.8,-1.4);
\draw[red] (3.1,-1.5) node{$\ell_1$};
\draw (5.1,2.5) node{$\ell_1: \delta = -2\epsilon$};

\draw[thick,->] (-2.5,-2.5) -- (2.5,-2.5);
\draw (-2.5,-2.8) node{\tiny $-1$};
\draw (0,-2.8) node{\tiny $0$};
\draw (2.5,-2.8) node{\tiny $+1$};
\draw (2.9,-2.5) node{$\delta$};

\draw[thick,->] (-2.5,-2.5) -- (-2.5,2.5);
\draw (-2.8,-2.5) node{\tiny $-1$};
\draw (-2.8,0) node{\tiny $0$};
\draw (-2.8,2.5) node{\tiny $+1$};
\draw (-2.5,2.9) node{$\epsilon$};

\draw (-1,1) node{$\stackrel{\delta}{\longrightarrow}$};
\draw (-0.4,1.2) node{$\big\uparrow \epsilon$};

\draw (1,0.6) node{$\stackrel{\delta}{\longrightarrow}$};
\draw (1.6,0.8) node{$\big\downarrow \epsilon$};

\draw (0.4,-1.4) node{$\stackrel{\delta}{\longleftarrow}$};
\draw (1.0,-1.2) node{$\big\downarrow \epsilon$};

\draw (-1.6,-1.0) node{$\stackrel{\delta}{\longleftarrow}$};
\draw (-1.0,-0.8) node{$\big\uparrow \epsilon$};

\draw[blue,thick] (-0.8,2) arc (110:-120:2.2cm);
\draw[thick,->] (-1.1,-2) -- (-1.2,-1.95);
\end{tikzpicture}
}
\end{center}
\caption{Visualization of the proof of Theorem~\ref{thm:strongfail}}\label{fig:cycling}
\end{figure}

Next, notice from Equations~\eqref{eqn:deltadot}-\eqref{eqn:epsdot}
that if $\delta > -2\epsilon$ then $\dot{\delta} > 0$, and if $\delta < -2\epsilon$
then $\dot{\delta} < 0$. Similarly, if $\epsilon > 2\delta$ then $\dot{\epsilon} > 0$, and if
$\epsilon < 2\delta$ then $\dot{\epsilon} < 0$. So, the range $\Omega'$ of $(\delta,\epsilon)$
values is partitioned into four (not axis-parallel) quadrants (see Figure~\ref{fig:cycling}), where at the
boundaries either $\dot{\delta} = 0$ or $\dot{\epsilon} = 0$.  Each
of these quadrants contains exactly one line segment along which $|\delta|=|\epsilon|$; we
will refer to these segments as ``arms."

We will use the above properties of the derivatives in the quadrants to show that the 
trajectory passes in finite time from each arm to the next (clockwise).  This will in turn imply
that the trajectory starting from any $(\delta(0),\epsilon(0))\notin C_{\rho}$ passes 
infinitely often through the four arms, and hence, in light of our earlier observation
that $\max\{|\delta(t)|,|\epsilon(t)|\}\rightarrow 1$ as $t\rightarrow\infty$,
there exists an infinite sequence of times $0\le s_1 < s_2 < s_3 < \cdots$ such that, along this sequence, 
both $|\delta(s_j)|\rightarrow 1$ and $|\epsilon(s_j)|\rightarrow 1$ as $j\rightarrow\infty$.
This implies directly that $\lim [ H(p(s_j))+ H(q(s_j)) ] = 0$. Furthermore, we see from~\eqref{eqn:deltadot}-\eqref{eqn:epsdot} that $|\dot \delta|, |\dot \epsilon| \leq 3/4$, so
outside of $C_\rho$ there is a lower bound of $(4/3)\sqrt{2\rho}$ on any interval $s_{i+1}-s_i$. These two 
conclusions imply the theorem.

Figure~\ref{fig:simulation} gives a numerical simulation of a typical trajectory,
showing the outward-spiraling behavior. 

\begin{figure}[ht]
\begin{center}
\includegraphics[height=3in,angle=0]{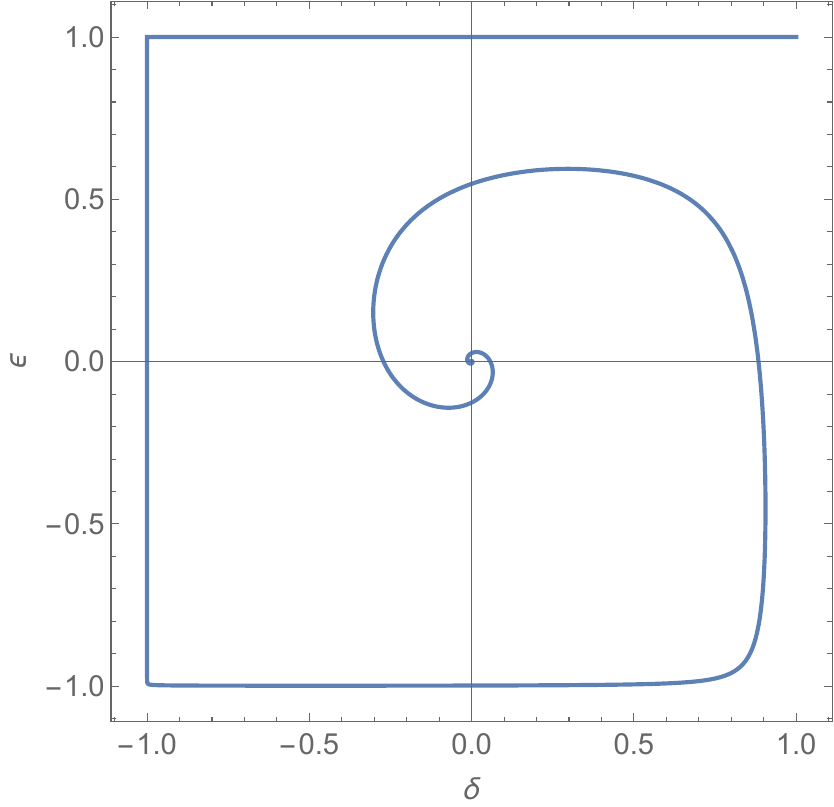} \quad
\includegraphics[height=3in,angle=0]{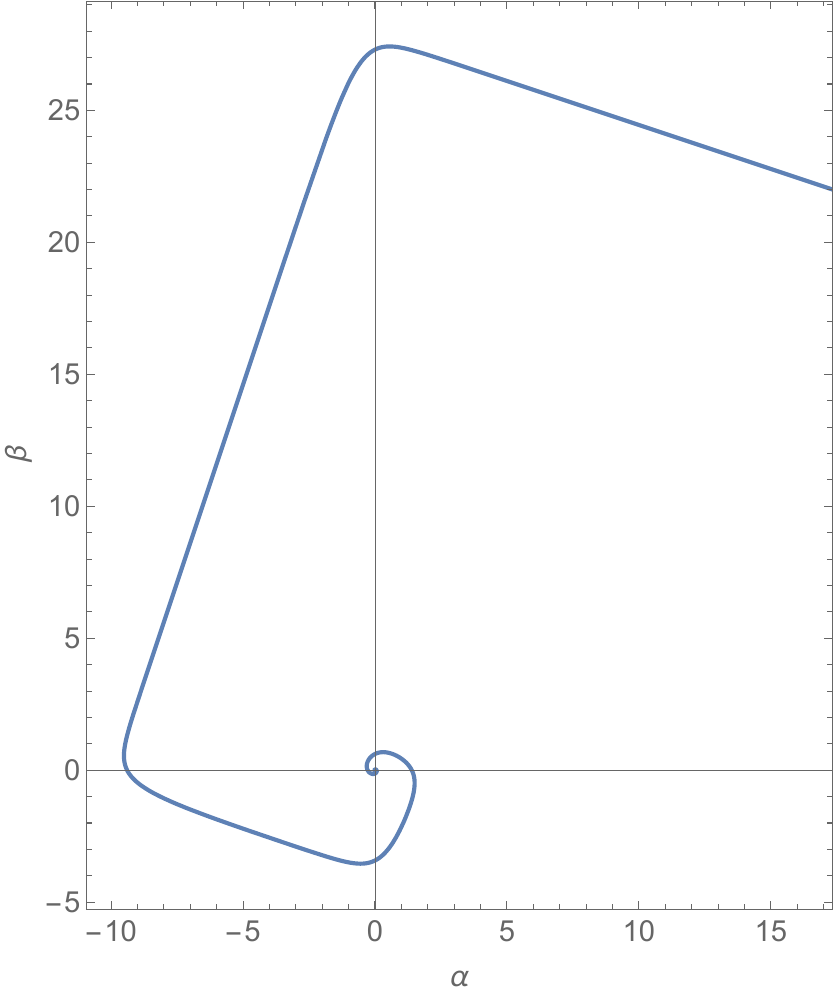}
\caption{A trajectory in $S$ for $t\in[0,150]$ for the payoff matrix~\eqref{nicequadgame},
              from initial condition $\delta=\delta_1 = \delta_2 = -10^{-6}$, $\epsilon=\epsilon_1 = \epsilon_2 = 0$. Plotted in terms of $(\delta,\epsilon)$ (left) and $(\alpha,\beta)$ (right). Each long nearly-straight part of the trajectory, on the right, corresponds to a traversal through one of the corners, on the left. We can conclude from this for example that the times spent in the corners increase roughly geometrically.}
\label{fig:simulation}
\end{center}
\end{figure}

We now prove the claim that the trajectory moves in finite time 
from one arm to the next (clockwise).  We prove this for just one adjacent pair of arms,
namely the ``north-west" and ``north-east" arms;
the same holds for the other three pairs by $C_4$ symmetry.

\begin{claim}\label{claim:cycling}
Suppose the trajectory starts at a point $a:=(-\alpha,\alpha)\notin C_\rho$ with $\alpha\in(0,1)$, i.e.,
$\delta(0)=-\epsilon(0)=-\alpha$.  Then at some finite time~$t$ the trajectory reaches
a point $b:=(\delta(t),\epsilon(t))=(\beta,\beta)$ with $\beta=\beta(\alpha)\in(0,1)$.
\end{claim}

\begin{proof}
The proof proceeds in two phases: first we show that the trajectory reaches a point $c:=(\gamma,2\gamma)$
with $\gamma\in(0,1/2)$ in finite time~$t'=t'(\alpha)$; then we
show that the trajectory starting from~$c$ reaches~$b$ in a further finite time~$t''=t''(\gamma)$.
Setting $t=t'+t''$ then proves the Claim.

For the first phase, we observe from~\eqref{eqn:deltadot} that the trajectory starting at $(-\alpha,\alpha)$
has $\dot{\delta} = \frac{1}{4}(1-\delta^2)(\delta+2\epsilon) \ge \frac{\alpha}{4}(1-\delta^2)$, since
$\delta\ge-\alpha$ and $\epsilon\ge \alpha$ throughout this phase.  Solving this differential equation
yields
\begin{equation}\label{eqn:desoln}
   \delta(t)\ge \frac{Ce^{\alpha t/2}-1}{Ce^{\alpha t/2}+1},  \qquad\hbox{\rm where $\displaystyle C=\frac{1+\alpha}{1-\alpha}$}. 
\end{equation}   
Since the maximum possible value of~$\delta$ on~$\ell_2$ is~$\frac{1}{2}$,
the trajectory must reach a point on~$\ell_2$ within finite time $t'\le \frac{2}{\alpha}\ln\left(\frac{3(1-\alpha)}{1+\alpha}\right)$.
Similarly, from~\eqref{eqn:epsdot}, the evolution of~$\epsilon$ during this phase satisfies
$\dot{\epsilon} = -\frac{1}{4}(1-\epsilon^2)(2\delta-\epsilon) \le\frac{1+2\alpha}{4}(1-\epsilon^2)$, with a solution of the form $$
   \epsilon(t)\le \frac{C'e^{(1+2\alpha)t/2}-1}{Ce^{(1+2\alpha)t/2}+1} \qquad\hbox{\rm for some $C'>0$}.  $$
Thus, at time~$t'$ as above, $\epsilon$~takes some value $2\gamma<1$, ensuring that the trajectory reaches~$\ell_2$
at the point $(\gamma,2\gamma)$. Since $\dot \epsilon \geq 0$ during this phase, $2\gamma\geq \alpha$. 

For the second phase, we note that the evolution of~$\delta$ satisfies 
$\dot{\delta} = \frac{1}{4}(1-\delta^2)(\delta+2\epsilon) \ge \frac{3\alpha}{8}(1-\delta^2)$, since
here $\epsilon \geq \delta \geq \gamma \geq \alpha/2$. 
Moreover, since $\dot \epsilon\leq 0$ during this phase, and $\epsilon \geq \delta$, we have $\delta\leq 2\gamma$ throughout the phase, so $\dot \delta$ is bounded away from $0$, hence 
 the trajectory reaches the arm $\delta=\epsilon$ in finite time~$t''(\gamma)$. 
Moreover, it does so at a point $b=(\beta,\beta)$ with $\beta\le2\gamma<1$. 
\end{proof}

To conclude the desired property that a trajectory starting from any $(\delta(0),\epsilon(0))\notin C_{\rho}$
passes infinitely often through all four arms, we simply apply Claim~\ref{claim:cycling} (taking advantage of the $C_4$ symmetry) repeatedly.  For the initial portion of the trajectory, whose starting point
$(\delta(0),\epsilon(0))$ may not lie on an arm, it should be clear that the arguments in the proof of
Claim~\ref{claim:cycling} also imply that the trajectory reaches a point $(\alpha,\pm\alpha)$ on 
the closest arm (clockwise), with $\alpha\in(0,1)$, in finite time.

This completes the proof of the theorem for trajectories starting in $S\setminus C_\rho$, 
for any $\rho>0$, and hence for trajectories starting in $S\setminus\{0\}$.
For the full proof, i.e., for trajectories starting in $(0,1)^4 \setminus \check S$,
we will show that any trajectory starting outside~$\check S$
eventually behaves in a similar manner to that described above for trajectories
in~$S\setminus\{0\}$. 

\begin{remark}\label{rem:skewsymm}
For trajectories in the skew-symmetric subspace~$\check S$, we
see from~\eqref{eqn:deltadot}-\eqref{eqn:epsdot} that $\dot\delta = -\frac{1}{4}(1-\delta^2)\delta$
and $\dot\epsilon = -\frac{1}{4}(1-\epsilon^2)\epsilon$, so the dynamics in~$\delta$ and~$\epsilon$
are independent and it is immediately clear that all such trajectories are attracted to the origin.
Thus inequality~\eqref{eq:entcond1} actually does hold in the set~$\check S$ of measure zero.
\end{remark}

\subsection{The dynamics in $(0,1)^4 \setminus \check S$} \label{subsec:subspaceScheck}
Consider now the full four-dimensional domain~$\Omega$.
To analyze trajectories outside~$\check S$, we make use of the potential functions $\mu_A,\mu_B$
defined in~\eqref{eqn:mudef}, whose logarithmic derivatives are given in~\eqref{eqn:Dmu}.
We focus on~$\mu_A$; a symmetrical analysis will then apply to~$\mu_B$.  Note from~\eqref{eqn:Dmu}
that $\logderiv\mu_A\le 0$ with equality only at the boundary values $\delta_1=\delta_2=\pm 1$.  
This implies that $\mu_A$ is monotonically (weakly) decreasing and, since it is bounded below (by~0),
it must tend to a limit $\mu_\infty\ge 0$.  We will show:

\begin{lemma} $\mu_\infty=0$.  \end{lemma} 
\begin{proof} Consider the expression
\begin{equation}\label{eqn:muint}
  \log \mu_\infty = \log\mu_A(0) + \int_0^\infty \logderiv \mu_A(t) dt = \log\mu_A(0) -\frac{1}{2} \int_0^\infty (1-\delta_1(t)\delta_2(t)) dt.  
\end{equation}  
If the right-hand integral in~\eqref{eqn:muint} is unbounded, then necessarily $\mu_\infty=0$ as desired.

We show that the integral is indeed unbounded. Suppose for contradiction that its value is $c<\infty$.
For $\eta\in(0,1)$ let
\begin{align*} \Gamma_\eta & = \{(\delta_1,\delta_2): [\delta_1,\delta_2\geq 1-\eta]
 \vee [\delta_1,\delta_2\leq -1+\eta] \} ; \\
   T_\eta &= \{t\ge 0: (\delta_1(t),\delta_2(t))\notin \Gamma_\eta\}. \end{align*}
Thus $T_\eta$ is the set of times at which the $(\delta_1,\delta_2)$ projection of the dynamics is outside
both of the $\ell_\infty$ balls of radius~$\eta$ around the corners $\delta_1=\delta_2=\pm1$.
For all $t\in T_\eta$, the integrand in~\eqref{eqn:muint}
satisfies $1-\delta_1(t)\delta_2(t)\ge \eta$.  
Now note that for any $t_0\in T_\eta$, $t_0+t \in T_{\eta/2}$
for $0\le t\le 1/3$.
 This follows from the
dynamical equation~\eqref{eqn:deltadot}, which implies that 
$| \dot \delta_i | \le\frac{3}{4}(1-\delta_i^2)$, and thus
$\delta_i$ cannot move from $1-\eta$ to $1-\eta/2$
(respectively, from $-1+\eta$ to $-1+\eta/2$) in less than time $1/3$.

Now if $T_\eta \neq \emptyset$, let $t_1=\inf T_\eta$ (e.g., $t_1=0$ if the dynamics starts outside~$\Gamma_\eta$).  By the observations above, the time interval
$[t_1,t_1+1/3)$ contributes at least $c'\eta \; (c'>0)$ to the integral. Next if $T_\eta \backslash [t_1,t_1+1/3) \neq \emptyset$, 
let $t_2=\inf T_\eta \backslash [t_1,t_1+1/3)$; the interval $[t_2,t_2+1/3]$ again contributes
a further $c'\eta$ to the integral.  
Continuing in this fashion, 
there can be at most $\frac{c}{c'\eta}$ such contributions.  This implies that, after some finite time~$\hat t$,
the dynamics never again leaves~$\Gamma_\eta$.  

Assume w.l.o.g.\ that
$\delta_1,\delta_2\ge 1-\eta$ for all $t\ge\hat t$. (A similar argument will apply if $\delta_1,\delta_2\le -1+\eta$  for all $t\ge\hat t$.) Note from~\eqref{eqn:epsdot} that, at all such times, $$
   \dot\epsilon_i = -\frac{1}{4}(1-\epsilon_i^2)(\delta_1+\delta_2-\epsilon_{3-i}) \le -\frac{1}{4}(1-\epsilon_i^2)(1-2\eta).  $$
Thus $\epsilon_1,\epsilon_2$ will both decrease monotonically and eventually become less than $-\frac{3}{4}$
by a finite time~$\hat{t}'$.  But in that case \eqref{eqn:deltadot} gives 
$$
   \dot{\delta}_i  =  \frac 1 4 (1 - \delta_i^2) (\delta_{3-i} + \epsilon_1 + \epsilon_2) \le -\frac{1}{8}(1 - \delta_i^2),    $$
implying that the $\delta_i$ decrease monotonically after $\hat{t}'$.  Thus in finite time one of the $\delta_i$ drops below $1-\eta$
and the trajectory exits~$\Gamma_\eta$, a contradiction.  
Hence we conclude that  our initial assumption that the integral in~\eqref{eqn:muint} is bounded must be false,
and therefore $\mu_\infty=0$.   
\end{proof}

An analogous analysis of the function~$\mu_B(t)$ (which is identical to~$\mu_A(t)$ with the $\delta_i$ replaced
by~$\epsilon_i$) shows that $\mu_B$ is also decreasing and has limit~$0$.
It follows that both $\delta_1(t)-\delta_2(t)$ and $\epsilon_1(t)-\epsilon_2(t)$ tend to~0
as $t\to\infty$.  Moreover, 
since the logarithmic derivatives $\logderiv(\delta_1(t)-\delta_2(t))$ and $\logderiv(\epsilon_1(t)-\epsilon_2(t))$
are both bounded in absolute value (by 3/2), the quantities $\delta_1(t)-\delta_2(t)$ and 
$\epsilon_1(t)-\epsilon_2(t)$ each remain positive, zero or negative throughout the trajectory. By the symmetries of the system we can w.l.o.g.\ suppose both are nonnegative, and we continue to assume so throughout this section.
Thus we have shown:
\begin{lemma} \label{closetoS}
For any $\eta>0$ and any initial point with $\delta_1\geq \delta_2$ and $\epsilon_1\geq \epsilon_2$,
 there is a time $t_0$ such that, for all $t\geq t_0$, the trajectory is confined to the region $\Delta_\eta$ defined by the inequalities 
\begin{equation}\label{eqn:closetoS}
    \delta_2 \leq \delta_1 \leq \delta_2 +\eta\qquad\hbox{\rm and}\qquad 
    \epsilon_2 \leq \epsilon_1 \leq \epsilon_2 +\eta.
\end{equation} \end{lemma} 
Now define
\begin{align*}
 \ltwob &:= \frac12((\delta_1+\delta_2)^2+(\epsilon_1+\epsilon_2)^2) \end{align*}
and note that since we are outside~$\check S$, $\ltwob >0$.
Equations~\eqref{eqn:deltadot} and~\eqref{eqn:epsdot} give
\begin{align} \dot{\ltwob} &= \frac14 [(1 - \delta_1 \delta_2) (\delta_1 + \delta_2)^2 + (1 - \epsilon_1 \epsilon_2) (\epsilon_1 + \epsilon_2)^2 -
 (\delta_1 + \delta_2) (\epsilon_1 + \epsilon_2) (\delta_1^2 + \delta_2^2 - \epsilon_1^2 - \epsilon_2^2)] .\label{dtwob} \end{align}

\begin{lemma} If $\eta\leq \frac{1}{3\sqrt{2}}$ and
$\ltwob\leq 1/2$ then $\dot{\ltwob} > 0$. \label{ltwobD} \end{lemma}
\begin{proof} Write $d=(\delta_1-\delta_2)/2,D=(\delta_1+\delta_2)/2,e=(\epsilon_1-\epsilon_2)/2,E=(\epsilon_1+\epsilon_2)/2$. 
Recall from Lemma~\ref{closetoS} that $d,e\geq 0$. Also, by assumption, $0<D^2+E^2=\ltwob/2 \leq 1/4$. 

First suppose that $\max\{|\delta_1|,|\delta_2|,|\epsilon_1|,|\epsilon_2|\}\geq 1/\sqrt{2}$.  This implies that
\[ \max\{D^2,E^2\} \geq \frac{1}{\sqrt{2}}(\frac{1}{\sqrt{2}}-\eta) = \frac12 - \eta/\sqrt{2}. \] 
But then $\lambda=2(D^2+E^2) \geq 1-\sqrt{2}\eta=2/3$, contradicting the assumption on~$\lambda$.

Hence $\max\{|\delta_1|,|\delta_2|,|\epsilon_1|,|\epsilon_2|\}< 1/\sqrt{2}$.
Rewriting~\eqref{dtwob} we get
\begin{align*} \dot{\ltwob} &= (1 - \delta_1 \delta_2) D^2 + (1 - \epsilon_1 \epsilon_2) E^2 -
 DE (\delta_1^2 + \delta_2^2 - \epsilon_1^2 - \epsilon_2^2) \\
 & > D^2/2+E^2/2 - |DE| \qquad\text{(using $\max\{|D|,|E|\}>0$)} 
 \\ & \geq 0. \label{dtwob2} \qedhere \end{align*}
 
\end{proof}
\begin{corollary} For any initial point in $(0,1)^4 \setminus \check S$ there is a finite $t_0$ such that 
for all $t\geq t_0$, $\ltwob(t)\geq \min\{1/2, \ltwob(t_0)\}=: \ltwob_0>0$. \end{corollary}

Our next goal is to establish that the potential function~$\pi$ defined in~\eqref{eqn:pidef} is eventually monotone decreasing.
(Compare the situation with the dynamics inside $S$, where $\pi$ is \emph{always} decreasing.) 
Applying Lemma~\ref{closetoS}, let $t_0$ be sufficiently large that the trajectory is confined in $[t_0,\infty)$ to $\Delta_{\sqrt{\frac12 \ltwob_0}}$.

Recalling the expression~\eqref{eqn:logderivpi} for the logarithmic
derivative of~$\pi$, and applying~\eqref{eqn:closetoS} which gives $(\delta_1-\delta_2)^2,(\epsilon_1-\epsilon_2)^2\le \frac12 \ltwob_0$, we 
 see that in $[t_0,\infty)$ 
\begin{align*}
   \logderiv\pi &= -(\delta_1\delta_2+\epsilon_1\epsilon_2) \\ &= 
   -\frac12 \ltwob + \frac14((\delta_1-\delta_2)^2+(\epsilon_1-\epsilon_2)^2) \\
   &\leq -\frac12 \ltwob_0 + \frac14 \ltwob_0 = -\frac14 \ltwob_0.
   \end{align*}
This shows that $\pi(t)\to 0$ as $t\to\infty$, as claimed. 

Consequently $\max\{|\delta_1(t)|,|\delta_2(t)|,|\epsilon_1(t)|,|\epsilon_2(t)|\}\to 1$.  To
complete the argument, we need to show that for all $t$ and $\eta_0>0$ there is a time $s>t$ s.t.\ $\min\{|\delta_1(t)|,|\delta_2(t)|,|\epsilon_1(t)|,|\epsilon_2(t)|\}\geq 1-\eta_0$. The first step is to apply Lemma~\ref{closetoS} so that we are at a time $t_0\geq t$ after which the trajectory is confined to $\Delta_{\eta}$ for
$\eta=\min\{c\eta_0,\sqrt{\ltwob_0}/16\}$ for a small constant $c>0$ (that will emerge from the argument).

We have that $\ltwob(t_0)\geq \ltwob_0$ so that  $\max\{|\delta_1(t_0)|,|\delta_2(t_0)|,|\epsilon_1(t_0)|,|\epsilon_2(t_0)|\}\geq \frac12 \sqrt{\ltwob_0}$. By the $C_4$ symmetry (and our symmetry-breaking assumption $\epsilon_1\geq \epsilon_2$), we may assume that $\epsilon_1(t_0)\geq \max\{|\delta_1(t_0)|,|\delta_2(t_0)|\}$ and 
$\epsilon_1(t_0) \geq \frac12 \sqrt{\ltwob_0} > \frac{1}{16} \sqrt{\ltwob_0}\ge\eta$.  As in the analysis of the 2-dimensional dynamics in~$S$, all that remains to show is that within finite time after $t_0$, 
the trajectory achieves $\epsilon_1=\delta_1$. 
We proceed to mimic the proof of Claim~\ref{claim:cycling} by decomposing progress into phases. But
first we note that at all times
\be 2\ltwob_0\leq (2\delta_1+\eta)^2+(2\epsilon_1+\eta)^2. \label{delarge} \ee
\par\medskip\noindent
Phase I: $-\epsilon_1\leq \delta_1 \leq 0$. 
During this phase we have
$$
  \dot\epsilon_1 = -\frac{1}{4}(1-\epsilon_1^2)(\delta_1+\delta_2-\epsilon_2) 
\geq \frac{1}{4}(1-\epsilon_1^2)(\epsilon_1-2\delta_1-\eta) 
\geq \frac{1}{4}(1-\epsilon_1^2)(\epsilon_1-\eta) >0,
$$
where the first inequality uses the fact that $\epsilon_2\ge \epsilon_1-\eta$ and $\delta_2\le\delta_1$.
The last inequality holds because, by assumption, at time $t_0$ we have $\epsilon_1(t_0) > \eta$, 
and $\epsilon_1$ will continue to increase at least until the condition $\delta_1\le 0$ is violated.

Now we consider~$\dot\delta_1$ during this phase.  We have
$$
   \dot\delta_1 = \frac{1}{4}(1-\delta_1^2)(\epsilon_1+\epsilon_2+\delta_2)
\geq \frac{1}{4}(1-\delta_1^2)(2\epsilon_1+\delta_1-2\eta) 
 \geq \frac{1}{4}(1-\delta_1^2)(\epsilon_1-2\eta)>0,
$$
where we have used the fact that throughout this phase $\epsilon_1+\delta_1 \ge 0$ and
also $\epsilon_1>2\eta$ (which holds at time~$t_0$ and thus throughout the phase since $\dot \epsilon_1>0$). 
In fact,  we have that $\epsilon_1\geq \frac12 \sqrt{\ltwob_0}$ and $\eta \leq \sqrt{\ltwob_0}/16$
and $\delta_1$ is always at least $\delta_1(t_0) > -1$, so $\dot\delta_1 > \frac{3}{32} (1-(\delta_1(t_0))^2) \sqrt{\ltwob_0}$.
Therefore, the phase terminates in finite time.

\par\medskip\noindent
Phase II:  $0\leq \delta_1 \leq  \min\{(1+2\eta)/2,\epsilon_1\}$.
We argue this phase has bounded duration as follows. First, we argue that while $\delta_1\le 2\eta$,
$\epsilon_1$ is non-decreasing. This is because at the start of the phase 
$\epsilon_1\ge \epsilon_1(t_0)\ge \frac12 \sqrt{\ltwob_0}\ge 8\eta$. Therefore, 
$\dot\epsilon_1\ge \frac{1}{4}(1-\epsilon_1^2)(\epsilon_1-2\delta_1-2\eta) \ge \frac 1 2 (1-\epsilon_1^2)\eta\ge 0$. 

Next, while $\epsilon_1\ge 2\eta$ we have that
$\dot\delta_1\ge \frac{1}{4}(1-\delta_1^2)(2\epsilon_1+\delta_1-2\eta)\ge \frac{1}{2}(1-\delta_1^2)\eta > 0$.
We will shortly argue that $\epsilon_1 > 2\eta$ for the duration of the phase. Therefore, there
is a lower bound on the rate of increase of $\delta_1$, so consequently the target value 
$\min\{(1+2\eta)/2,\epsilon_1\}$ is reached in finite time; moreover, at the end of the phase, 
$\epsilon_1$ is positive and bounded away from $1$. 

To see that $\epsilon_1 > 2\eta$ for the duration of the phase, let $t$ be the earliest time after the start of 
the phase where $\epsilon_1(t)\le 2\eta$ (so, in fact, $\epsilon_1(t) = 2\eta$). Then, there must have been 
a time $t' < t$ in the phase at which $\delta_1(t') > 2\eta$, because otherwise $\epsilon_1$ does not decrease
and its initial value is greater than $2\eta$. Assume for contradiction that the phase does not end before
time $t$. As $\epsilon_1\ge 2\eta$ during the entire interval $[t',t]$, $\delta_1$ must increase during this
entire interval. We get that $\delta_1(t) > 2\eta\ge \epsilon_1(t)$, in contradiction to the condition that in
the phase $\delta_1\le\epsilon_1$. Thus, the phase must end before time~$t$.

\par\medskip\noindent
Phase III: $ (1+2\eta)/2 \leq \delta_1 \leq \epsilon_1$. 
Here $\epsilon_1$ is decreasing:
 $$
  \dot\epsilon_1 = -\frac{1}{4}(1-\epsilon_1^2)(\delta_1+\delta_2-\epsilon_2) 
\le -\frac{1}{4}(1-\epsilon_1^2)(2\delta_1-\epsilon_1-\eta) 
\le -\frac{1}{4}(1-\epsilon_1^2)(1+2\eta-1-\eta) 
< 0.$$

On the other hand, $\delta_1$ is still increasing:
\begin{align}
   \dot\delta_1 &= \frac{1}{4}(1-\delta_1^2)
    (\epsilon_1+\epsilon_2+\delta_2) 
 \geq \frac{1}{4}(1-\delta_1^2)(2\epsilon_1+\delta_1-2\eta) 
 \geq  \frac{1}{4}(1-\delta_1^2)(\frac32 + \eta) > 0.
 \nonumber 
\end{align}
It follows that Phase~III continues until the desired event $\delta_1=\epsilon_1$, and the uniform lower bounds
on the above derivatives again ensure that this happens in finite time.  
This completes the proof of Theorem~\ref{thm:strongfail}.  \hfill\qed

\begin{remark}
It is not hard to check that the analysis of this section applies equally to the same game 
with the value $z=\frac{1}{2}$ replaced by any $z\in(0,1)$.
\end{remark}

\section{Proof of Theorem~\ref{thm:main}} \label{sec:thm2proof}
We recall the form of the replicator dynamics given in Equations~\eqref{eq:replicator4pprelims} and~\eqref{eq:replicator4qprelims}, 
which we restate here for convenience:
\begin{eqnarray}
\dot p_{i,b}  &=&   p_{i,b}(1 -  p_{i,b})\cdot \sum_{\substack{x\in\{0,1\}^n \\ x_i=b}} \sum_{y\in\{0,1\}^m} K_{A,i}(x,y,p,q) \bigl( u(x,y) - u(\bar x^i,y)\bigr) \label{eq:replicator4p}\,;\\
\dot q_{j,b}  &=&  -q_{j,b}(1 -  q_{j,b})\cdot \sum_{x\in\{0,1\}^n}\sum_{\substack{y\in\{0,1\}^m \\ y_j=b}}  K_{B,j}(x,y,p,q) \bigl(u(x,y)-u(x,\bar y^j)\bigr) \label{eq:replicator4q}\,,
\end{eqnarray}
where
\begin{eqnarray*}
K_{A,i}(x,y,p,q) &=& \prod_{i'\ne i} p_{i',x_{i'}} \prod_j q_{j,y_j}\,; \\
K_{B,j}(x,y,p,q) &=& \prod_{i} p_{i,x_i} \prod_{j'\ne j} q_{j',x_{j'}}\,.
\end{eqnarray*}
Throughout this section, we also let $U := \max_{x\in\{0,1\}^n,y\in\{0,1\}^m} |u(x,y)|$ denote the maximum (in absolute value) of any utility
value in the game.
For $\gamma\in \left(0,\frac 1 2\right)$,
we say that a state $(p,q)$ is {\it $\gamma$-near\/} the corner $(x,y)$ iff for all $i$,
$p_{i,x_i} > 1 - \gamma$ and for all $j$, $q_{j,y_j} > 1 - \gamma$. If
$(p,q)$ is not $\gamma$-near $(x,y)$, we say that it is {\it $\gamma$-far
from\/} $(x,y)$. Notice that, as $\gamma < \frac 1 2$, $(p,q)$ cannot be 
$\gamma$-near more than one corner.  We begin with a useful claim and a useful
lemma about the behavior of the dynamics near corners, both of which we will use
repeatedly in the proof.
\begin{claim}\label{cl:speed near corner}
Let $(x,y)\in\{0,1\}^{n+m}$ be a corner, and let $\alpha > 0$. There exists 
$\hat{\gamma}\in \left(0,\frac 1 2\right)$ and constants $C_1 > C_0 > 0$ such 
that the following holds. For every $\gamma\le\hat{\gamma}$, if $(p,q)$ is 
$\gamma$-near $(x,y)$, then for every $i\in [n]$ it holds that:
\begin{enumerate}
\item if $u(x,y) - u(\bar x^i,y) \ge\alpha$ then $\dot p_{i,x_i} > C_1\cdot (1-p_{i,x_i}) = C_1 \cdot p_{i,\overline{x_i}}\,$\,;
\item[] if $u(x,y) - u(x,\bar y^j) \le -\alpha$ then $\dot q_{j,y_j} > C_1\cdot (1-q_{j,y_j}) = C_1 \cdot q_{j,\overline{y_j}}\,$; 
\item if $u(x,y) - u(\bar x^i,y) \le-\alpha$ then $\dot p_{i,x_i} < -C_1\cdot (1-p_{i,x_i}) = -C_1 \cdot p_{i,\overline{x_i}}\,$\,;
\item[] if $u(x,y) - u(x,\bar y^j) \ge \alpha$ then $\dot q_{j,y_j} < -C_1\cdot (1-q_{j,y_j}) = -C_1 \cdot q_{j,\overline{y_j}}$\,;
\item if $u(x,y) - u(\bar x^i,y) = 0$ then $|\dot p_{i,x_i}| < C_0\cdot (1-p_{i,x_i}) = C_0 \cdot p_{i,\overline{x_i}}\,$\,;
\item[] if $u(x,y) - u(x,\bar y^j) = 0$ then $|\dot p_{i,x_i}| < C_0\cdot (1-p_{i,x_i}) = C_0 \cdot p_{i,\overline{x_i}}\,$\,.
\end{enumerate}
\end{claim}
\par\noindent
{\bf Note:} {\it We remind the reader that $\overline{x_i} = 1-x_i$ denotes the flipped value of the $i$th coordinate of~$x$,
whereas $\bar{x}^i$, as defined after~\eqref{eq:replicator4p2prelims}, is the vector $x\in\{0,1\}^n$ with its $i$th coordinate flipped.}

\begin{proof}
Notice that as $(p,q)$ is $\gamma$-near $(x,y)$, and hence $\hat{\gamma}$-near
$(x,y)$, then $K_{A,i}(x,y,p,q) > (1-\hat{\gamma})^{m+n}$ for all~$i$, and for all
$(x',y')\ne (x,y)$ (except for the pair $(\bar x^i,y)$), 
$K_{A,i}(x',y',p,q)\le\hat{\gamma}$ for all~$i$. Therefore, by
Equation~\eqref{eq:replicator4p}, for $i$ such that $u(x,y) - u(\bar x^i,y) \ge\alpha$
we get that
\begin{equation}\label{eq:temp1}
\dot p_{i,x_i} > (1-\hat{\gamma})\cdot (1 -  p_{i,x_i})\cdot \left((1-\hat{\gamma})^{m+n}\alpha - 2^{m+n}U\hat{\gamma}\right).
\end{equation}
By an identical argument, for $i$ such that $u(x,y) - u(\bar x^i,y) \le -\alpha$,
\eqref{eq:temp1}~again holds with the inequality reversed and the right-hand side negated.
On the other hand, for $i$ such that $u(x,y) - u(\bar x^i,y) = 0$, we get that
$$
|\dot p_{i,x_i}| \le (1 -  p_{i,x_i})\cdot 2^{m+n}U\hat{\gamma}.
$$
For a sufficiently small $\hat{\gamma} = \hat{\gamma}(m,n,U,\alpha) > 0$, there are
constants $C_1 > C_0 > 0$ such that
$$
(1-\hat{\gamma})\cdot \left((1-\hat{\gamma})^{m+n}\alpha - 2^{m+n}U\hat{\gamma}\right) \ge C_1 > C_0 >
2^{m+n}U\hat{\gamma}.
$$
This completes the proof for $i\in[n]$.  The proof for $j\in[m]$ is analogous.
\end{proof}

\begin{lemma}\label{lm:min far duration}
If at time $t$, $(p(t),q(t))$ is $\gamma$-far from every corner, then during the
entire interval $\left[t,t+\frac{1}{4U}\right]$, the trajectory of the Red Queen dynamics
stays $\frac{\gamma}{2}$-far from every corner.
\end{lemma}

\begin{proof}
Let $(x,y)$ be any corner. From the hypothesis of  the lemma, we know that $(p(t),q(t))$ 
is $\gamma$-far from~$(x,y)$.
Without loss of generality, we may thus assume that $p_{1,x_1}(t)\le 1-\gamma$.
In order for the dynamics to later reach a point that is $\frac\gamma 2$-near $(x,y)$,
there must exist a time interval $[t',t'+\tau]$ with $t\le t'$ such that 
$p_{1,x_1}(t') = 1-\gamma$,  $p_{1,x_1}(t'+\tau) = 1-\frac \gamma 2$,  and for all $s\in (t',t'+\tau)$,
$p_{1,x_1}(s)\in (1-\gamma,1-\frac{\gamma}{2})$.

Consider the first such interval. Since all utility values lie in $[-U,U]$, and 
$p_{1,x_1}(s) \ge 1-\gamma$ throughout the interval, we have from
Equation~\eqref{eq:replicator4p} that $\dot{p}_{1,x_1}\le 2U\gamma(1-\gamma)$
over the interval.  Therefore, it must be the case that $$
   1-\frac\gamma 2 = p_{1,x_1}(t'+\tau)\le (1-\gamma) + 2U\gamma(1-\gamma)\tau. $$
Hence we conclude that 
$\tau\ge\frac{1}{4U(1-\gamma)}\ge\frac{1}{4U}$.
As this argument can be made for every corner, the lemma follows.
\end{proof}

We turn now to the proof of Theorem~\ref{thm:main}, starting with case~{\bf (i)}. 
The condition that there exists no (weak) pure Nash
equilibrium is equivalent to the following assertion. For every corner $(x,y)$ of
the feasible region (namely, $x\in\{0,1\}^n$ and $y\in\{0,1\}^m$), {\em at least one} of
the following two conditions holds for some $\alpha>0$:
\begin{enumerate}
\item there exists $i$ for which $u(x,y) - u(\bar{x}^i,y) \le -\alpha$;
\item there exists $j$ for which $u(x,\bar{y}^j) - u(x,y) \le -\alpha$. 
\end{enumerate} 
The union of these conditions expresses the fact that at least one player has an incentive to 
defect from~$(x,y)$.
Therefore, if $(p,q)$ is $\gamma$-near a corner $(x,y)$, for $\gamma\le\hat{\gamma}$
as in Claim~\ref{cl:speed near corner}, then case~2 of that claim ensures that, for
some constant $C>0$, 
\begin{equation}\label{eq:escape}
\bigl(\exists i:\, \dot p_{i,x_i} < -C\cdot (1-p_{i,x_i})\bigr)\vee 
\bigl(\exists j:\, \dot q_{j,y_j} < -C\cdot (1-q_{j,y_j})\bigr).
\end{equation}

Now, we divide the timeline $[0,\infty)$ into {\it epochs\/} $[t_k,t_{k+1})$, for $k=0,1,2,\dots$,
where $t_0 = 0$ and the other~$t_k$ are defined inductively as follows. If, at time $t_k$,
$(p(t_k),q(t_k))$ is $\gamma$-far from every corner, then we set $t_{k+1} = t_k+\frac{1}{4U}$. 
Otherwise, if $(p(t_k),q(t_k))$ is $\gamma$-near a corner $(x,y)$, then we set $t_{k+1}$ to be
the minimum $t > t_{k}$ such that $(p(t),q(t))$ is $\gamma$-far from~$(x,y)$.

The proof rests on the following claim.
\begin{claim}\label{cl:max near duration}
In the above scenario,
if $(p(t),q(t))$ is $\gamma$-near a corner, then there exists some finite $t'>t$
such that $(p(t'),q(t'))$ is $\gamma$-far from every corner.
\end{claim}

\begin{proof}
Suppose that $(p(t),q(t))$ is $\gamma$-near the corner $(x,y)$.
By~\eqref{eq:escape}, for $t'\ge t$, as long as $(p(t'),q(t'))$ is $\gamma$-near $(x,y)$,
there is at least one player that moves away from that corner.  Without loss of generality
we may assume this is player $i=1$ on team~A.
Thus, $\dot{p}_{1,x_1}(t') < -C\cdot (1-p_{1,x_1})$, or equivalently, 
$\dot{p}_{1,\overline{x_1}}(t') > C\cdot p_{1,\overline{x_1}}$ for all such~$t'$.

Therefore, if during the interval
$[t,t')$ the dynamics stays $\gamma$-near the corner 
$(x,y)$, then $p_{1,\overline{x_1}}(t') > p_{1,\overline{x_1}}(t) e^{C (t'-t)}$. 
However, for
$t'\ge t + \frac{\ln(\gamma/p_{1,\overline{x_1}}(t))}{C}$ we get that 
$p_{1,\overline{x_1}}(t') > \gamma$, so $(p(t'),q(t'))$ must be $\gamma$-far 
from the corner $(x,y)$ at time $t + \frac{\ln(\gamma/p_{1,\overline{x_1}}(t))}{C}$ 
or earlier.  Now let~$t'$ be the earliest time in $[t, t + \frac{\ln(\gamma/p_{1,\overline{x_1}}(t))}{C})$
at which either $p_{1,x_1}(t') = 1-\gamma$, or some other coordinate (in either team~A
or team~B) satisfies the analogous equality.  Since $\gamma < \frac 1 2$, it must be the case 
that at this time~$t'$, $(p(t'),q(t'))$ is $\gamma$-far from every corner.
\end{proof}

We are now ready to complete the proof of case {\bf (i)}. Claim~\ref{cl:max near duration}
implies that the number of epochs that begin at a state that is $\gamma$-far from every
corner is unbounded, because if an epoch does not satisfy this condition then it has
finite length and the following epoch will satisfy this condition. By Lemma~\ref{lm:min far duration},
in every epoch that begins at a state that is $\gamma$-far from every corner, the trajectory
remains $\frac{\gamma}{2}$-far from every corner for the duration of the epoch, which is $\frac{1}{4U}$.
But then, for such an epoch $\left[t,t+\frac{1}{4U}\right)$, for every 
$t'\in \left[t,t+\frac{1}{4U}\right)$ we
have $H(p(t')) + H(q(t'))\ge \frac \gamma 2 \ln\frac{2}{\gamma}$. Therefore,
taking $\varepsilon = \frac \gamma 4 \ln\frac{2}{\gamma}$, we get
$$
\int_t^{t+\frac{1}{4U}} \max\left\{0,H(p(t')) + H(q(t')) - \varepsilon\right\} dt'\ge 
    \frac{\gamma}{16U} \ln\frac{2}{\gamma}.
$$
Together with the fact that there are infinitely many disjoint epochs of this sort,
Property~\ref{propMamaBear} holds, as claimed.  This completes the proof of case~{\bf (i)}.

Next we consider case {\bf (ii)}. If the conditions of case {\bf (i)} are satisfied, then its stronger conclusion
is implied. Otherwise, we have that no corner is a strict pure Nash equilibrium, but there 
exists at least one corner $(x,y)$ with the following three properties:
\begin{enumerate}
\item for all $i$, $u(x,y) - u(\bar x^i,y) \ge 0$;
\item for all $j$, $u(x,\bar y^j) - u(x,y) \ge 0$;
\item\label{it:indifferent} At least one of the above inequalities holds with equality.
\end{enumerate}
Such a corner is a \emph{weak} Nash equilibrium, as no player has an
advantage defecting, but at least one player is indifferent to defecting (as expressed
by condition~3). The issue
in this case is that the dynamics might become ``trapped'' by
and approach such a weak equilibrium.  However, we will show 
that even if this is the case, the convergence is so slow that the cumulative 
entropy is unbounded.

Consider a weak pure Nash equilibrium corner $(x,y)$. Let $I_0,J_0$
be the sets of $i,j$, respectively, for which the indifference property~\ref{it:indifferent} above holds at $(x,y)$. 
We know that $I_0\cup J_0\ne\emptyset$.
Now cases~2 and~3 of Claim~\ref{cl:speed near corner} imply
that there exists $\gamma\in \left(0, \frac 1 2\right)$ and constants $C_1 > C_0 > 0$
such that if $(p,q)$ is $\gamma$-near $(x,y)$, then the following inequalities hold:
\begin{equation}\label{eq:speed}
\begin{array}{ll}
\dot p_{i,\overline{x_i}} < -C_1\cdot p_{i,\overline{x_i}}, & \forall i\not\in I_0\\
\dot p_{i,\overline{x_i}} > -C_0\cdot p_{i,\overline{x_i}}, & \forall i\in I_0\\
\dot q_{j,\overline{y_j}} < -C_1\cdot q_{j,\overline{y_j}}, & \forall j\not\in J_0\\
\dot q_{j,\overline{y_j}} > -C_0\cdot q_{j,\overline{y_j}}, & \forall j\in J_0.
\end{array}
\end{equation}

If $(x,y)$ is not a weak pure Nash equilibrium, then Equation~\eqref{eq:escape} holds, as
in case {\bf (i}). 
Note that as there are finitely many corners, we can choose 
$\gamma$ uniformly for all corners, so that $C_0,C_1$ are also chosen uniformly 
for all weak pure Nash corners, and $C$ for all other corners. From now on, we set $\gamma,C_0,C_1,C$ to
be these uniform values.

The rest of the proof follows along similar lines to that of case~{\bf (i)}, but with
some added technicalities to handle the weak pure Nash corners.  We again
divide the timeline into epochs $[t_k,t_{k+1})$, starting with $t_0 = 0$, as 
follows. If $(p(t_k),q(t_k))$ is $\gamma$-far from all corners, we set 
$t_{k+1} = t_k + \frac 1 {4U}$. If $(p(t_k),q(t_k))$ is $\gamma$-near a
non-Nash corner, we set $t_{k+1}$ to be the minimum $t > t_k$ at which
the dynamics is $\gamma$-far from that corner. If $(p(t_k),q(t_k))$
is $\gamma$-near a weak Nash corner $(x,y)$, we set 
$t_{k+1} = \min\{t_{k+1}^{(1)},t_{k+1}^{(2)}\}$, where $t_{k+1}^{(1)}$ is
the minimum $t > t_k$ such that the state of the dynamics is $\gamma$-far
from that that corner, and $t_{k+1}^{(2)}$ is defined as follows. Let
\begin{eqnarray*}
\gamma_{\max}(t) & = & \max\{\max_i\{p_{i,\overline{x_i}}(t)\}, \max_j\{q_{j,\overline{y_j}}(t)\}\} \label{eq:max idx}\,; \\
\gamma_0(t) & = & \max\{\max_{i\in I_0}\{p_{i,\overline{x_i}}(t)\}, \max_{j\in J_0}\{q_{j,\overline{y_j}}(t)\}\} \label{eq:max 0-idx}\,.
\end{eqnarray*}
Note that $\gamma_{\max}(t) \ge \gamma_0(t)$ for all~$t$.
If $\gamma_{\max}(t_k) = \gamma_0(t_k)$, then
let $t_{k+1}^{(2)}$ be the minimum $t > t_k$ such that $(p(t),q(t))$ 
is $\frac{\gamma_0(t_k)}{2}$-near $(x,y)$. Otherwise, let $t_{k+1}^{(2)}$ be the 
minimum $t > t_k$ such that $\gamma_{\max}(t) = \gamma_0(t)$.
Notice that $t_{k+1}^{(1)}$, $t_{k+1}^{(2)}$ could be infinite.
\begin{claim}\label{cl:zero}
If $(p(t_k),q(t_k))$ is $\gamma$-near a weak pure Nash corner $(x,y)$ and 
$\gamma_{\max}(t_k) > \gamma_0(t_k)$, then $t_{k+1}$ is finite.
\end{claim}

\begin{proof}
If $t_{k+1}^{(1)} < \infty$ then clearly $t_{k+1}$ is finite.
Otherwise, by Equations~\eqref{eq:speed}, if at time $t$ we have that
$\gamma_{\max}(t) > \gamma_0(t)$, then
$\dot \gamma_{\max}(t) < -C_1 \gamma_{\max}(t)$ and
$\dot \gamma_0(t) > -C_0 \gamma_0(t)$.
So, assume that for every $t'\in [t_k,t]$ it is the case that
$\gamma_{\max}(t') > \gamma_0(t')$. Then
$\gamma_{\max}(t) < \gamma_{\max}(t_k)\cdot e^{-C_1 (t-t_k)}$,
whereas 
$\gamma_0(t) > \gamma_0(t_k)\cdot e^{-C_0 (t-t_k)}$.
However, if we set
$$
t > t_k + \frac{\ln\frac{\gamma_{\max}(t_k)}{\gamma_0(t_k)}}{C_1-C_0},
$$
then we get $\gamma_0(t) > \gamma_{\max}(t)$, which is a contradiction.
\end{proof}

We are now ready to prove case~{\bf (ii)}. Clearly, if $(p(t_k),q(t_k))$ is 
$\gamma$-far from every corner for an infinite subsequence of $k$,
then the analysis of case {\bf (i)} holds and we are done. The only
alternative situation is that the dynamics gets ``trapped" by a weak pure Nash
corner.  In this case, by the proof of Claim~\ref{cl:zero} there is a finite integer~$k_0$
such that, for every $t\ge t_{k_0}$, $\gamma_{\max}(t) = \gamma_0(t)$.

We now identify two cases, according to whether the number of epochs is finite
or infinite.  In the first case, there are only finitely many epochs.  Suppose the last epoch is $[t_{k-1},t_k)$, 
with $k\ge k_0$, and consider the infinite time interval $[t_k,\infty)$.  Since the current epoch fails
to end in finite time, it must be the case that $\gamma_{\max}(t) > \frac{\gamma_{\max}(t_k)}{2}$
for all $t\ge t_k$, and hence also $H(p(t)) + H(q(t))\ge \frac{\gamma_{\max}(t_k)}{2}\ln\frac{2}{\gamma_{\max}(t_k)}$
for all $t\ge t_k$.  This immediately implies that  $\int_{t_k}^\infty \left(H(p(t))+H(q(t)\right) dt =\infty$, as desired.

In the second case, there are infinitely many epochs.  Consider any such epoch $[t_k,t_{k+1})$ with $k\ge k_0$.
We claim that, at any time $t$ in this epoch, for all $i\in I_0$ we have 
\begin{equation}\label{eq:new convergence to weak Nash}
   \dot p_{i,\overline{x_i}}(t) \ge -C_2\gamma_{\max}(t_k)^2,
\end{equation}
for some absolute constant $C_2>0$, with the same bound holding for
$\dot q_{j,\overline{y_j}}(t)$ for all $j\in J_0$.  To see this claim, note
from the dynamics~\eqref{eq:replicator4p} that $$
 \dot p_{i,\overline{x_i}}(t) =  p_{i,\overline{x_i}}(t)p_{i,x_i}(t)\cdot \sum_{\substack{x' \\ x'_i=\overline{x_i}}}\sum_{y'} K_{A,i}(x\myprime,y\myprime,p,q) 
     \bigl( u(x\myprime,y\myprime) - u(\bar{x\myprime}^{\mkern0.5mu i},y\myprime)\bigr).  $$
Now $p_{i,\overline{x_i}}(t)p_{i,x_i}(t) = p_{i,\overline{x_i}}(t)(1-p_{i,\overline{x_i}}(t)) \le \gamma_{\max}(t_k)-\gamma_{\max}(t_k)^2$,
since $p_{i,\overline{x_i}}(t) \le\gamma_{\max}(t_k)$ for all $t\ge t_k$ and the function $z(1-z)$ is
increasing for small positive~$z$.  For the sum over $(x\myprime,y\myprime)$, we observe the
following.  If $(x\myprime,y\myprime) = (x,y)$ then the contribution is $0$, as 
$u(x,y) - u(\bar{x}^i,y) = 0$ (since $i\in I_0$).  Otherwise, $K_{A,i}(x\myprime,y\myprime,p,q)$ 
contains at least one factor of the form $p_{i',\overline{x_{i'}}}$ or $q_{j',\overline{y_{j'}}}$,
which is at most $\gamma_{\max}(t_k)$, so $K_{A,i}(x\myprime,y\myprime,p,q)\le\gamma_{\max}(t_k)$.
Moreover, as always $u(x\myprime,y\myprime) - u(\bar{x\myprime}^{\mkern0.5mu i},y\myprime) \ge -2U$.
Putting all this together, and noting that the number of such terms is less than $2^{n+m}$, we deduce
that~\eqref{eq:new convergence to weak Nash} holds for a suitable constant~$C_2$.  An identical
argument proves the same bound for $\dot q_{j,\overline{y_j}}(t)$ for $j\in J_0$.

The bound in~\eqref{eq:new convergence to weak Nash} implies that $\gamma_{\max}$ cannot decrease
very fast during the epoch.  To see this, recall the definition of $\gamma_{\max}$ and the fact that, for all $t\ge t_{k_0}$,
$\gamma_{\max}$ is achieved by some $i\in I_0$ or $j\in J_0$.  Then by a standard generalized version of 
Rolle's Theorem (see, e.g., \cite{Rolle})
the uniform lower bound~\eqref{eq:new convergence to weak Nash} on the derivatives  
$\dot p_{i,\overline{x_i}}(t)$ and $\dot q_{j,\overline{y_j}}(t)$ for $i\in I_0$ and $j\in J_0$ ensures that
$\gamma_{\max}(t_k+t') \ge \gamma_{\max}(t_k) - t'C_2\gamma_{\max}(t_k)^2$.  Setting
$t'=\frac{1}{2C_2\gamma_{\max}(t_k)}$ implies that $\gamma_{\max}(t_k+\frac{1}{2C_2\gamma_{\max}(t_k)}) \ge \frac{\gamma_{\max}(t_k)}{2}$,
and hence, from the definition of $t_{k+1}^{(2)}$, we see that $t_{k+1} \ge t_k + \frac{1}{2C_2\gamma_{\max}(t_k)}$. 
Moreover, during the entire epoch
$[t_k,t_{k+1})$, we have that $H(p(t)) + H(q(t))\ge \frac{\gamma_{\max}(t_k)}{2}\ln\frac{2}{\gamma_{\max}(t_k)}$.
Hence,
\begin{eqnarray*}
          \int_0^\infty \left(H(p(t))+H(q(t)\right) dt 
& \ge & \sum_{k\ge k_0} \int_{t_k}^{t_{k+1}} \left(H(p(t))+H(q(t)\right) dt \\
& \ge & \sum_{k\ge k_0} \frac{1}{2C_2\gamma_{\max}(t_k)}\cdot \frac{\gamma_{\max}(t_k)}{2}\ln\frac{2}{\gamma_{\max}(t_k)} \\
& \ge & \sum_{k\ge k_0} \frac{1}{4C_2}  =  \infty.
\end{eqnarray*}
This completes the proof of case~{\bf (ii)} of Theorem~\ref{thm:main}. 

It remains to prove case~{\bf (iii)}.  Let $(x,y)$
be a strict pure Nash equilibrium, i.e., $x\in\{0,1\}^n$, $y\in\{0,1\}^m$, and for some $\alpha>0$
we have:
\begin{enumerate}
\item for all $i$, $u(x,y) - u(\bar x^i,y) \ge\alpha$; 
\item for all $j$, $u(x,\bar y^j) - u(x,y) \ge\alpha$.
\end{enumerate}
Thus, if we initialize the dynamics at a point $(p(0),q(0))$ that is $\gamma$-near the corner~$(x,y)$,
where $\gamma\le\hat\gamma$ as specified in Claim~\ref{cl:speed near corner}, then case~1
of that claim ensures that
\begin{equation}\label{eqn:pqlb}
\begin{array}{r@{}l}
 \dot p_{i,x_i} &{} > C\cdot p_{i,\overline{x_i}} , \forall i \,;\\
 \dot q_{j,y_j} &{} > C\cdot q_{j,\overline{y_j}}, \forall j\,,
\end{array}
\end{equation}
for a suitable constant $C>0$.
We can restate~\eqref{eqn:pqlb} as:
\begin{equation}\label{eqn:pqub}
\begin{array}{r@{}l}
 \dot p_{i,\overline{x_i}} &{} < -C\cdot p_{i,\overline{x_i}}, \forall i\,;\\
 \dot q_{j,\overline{y_j}} &{} < -C\cdot q_{j,\overline{y_j}}, \forall j\,.
\end{array}
\end{equation}
Thus we see that the trajectory of the dynamics approaches~$(x,y)$.
It remains to verify that the convergence is sufficiently rapid to prevent the 
cumulative entropy from diverging. This follows from the bounds~\eqref{eqn:pqub},
which imply that, for some absolute constants $c_i, d_j$, 
we have $p_{i,\overline{x_i}}(t) < c_i e^{-C t}$ and $q_{j,\overline{y_j}}(t) < d_j e^{-C t}$. Therefore,
\begin{eqnarray*}
          \int_0^\infty \left(H(p(t))+H(q(t)\right) dt 
& < & 2\sum_{i=1}^n \int_0^{\infty} p_{i,\overline{x_i}}(t)\log\frac{1}{p_{i,\overline{x_i}}(t)} dt +
          2\sum_{j=1}^m \int_0^\infty q_{j,\overline{y_j}}(t)\log\frac{1}{q_{j,\overline{y_j}}(t)} dt \\
& < & 2\sum_{i=1}^n \int_0^{\infty} c_i(C t - \log c_i) e^{-C t} dt +
          2\sum_{j=1}^m \int_0^{\infty} d_j (C t - \log d_j) e^{-C t} dt \\			
& < & \infty,
\end{eqnarray*}
where the first inequality used $\gamma < \frac 1 2$.  This completes the proof
of case~{\bf (iii)}, and therefore of Theorem~\ref{thm:main}.

\section{Discussion} \label{sec:discussion}
We close the paper with some final remarks and open questions.
\begin{enumerate}
\item
The combination of Theorem~\ref{thm:strongfail} and
Theorem~\ref{thm:main} shows that systems can undergo dramatic, and never-ending, oscillations in diversity. There are of course many simplifications in the mathematical model, but the phenomenon has a clear basis in the dynamics, since it illustrates the ability of a a species to emerge from a period in which it was poorly fit for the competitive pressures. It would be interesting to investigate whether any natural ecologies exhibit
this type of behavior.

\item In a finite population there is a lower bound on the genome frequencies actually present in the population. Therefore, if the dynamical model analyzed in this paper approaches a monoculture, then all diversity may vanish. In such a case, the species may be trapped in an inferior pure strategy, perhaps leading to its extinction, due to the absence of diversity. This indicates that the Red Queen
dynamics is a model that may be valid for only a finite time horizon.

\item  According to Fisher's fundamental theorem of natural selection~\cite{Fisher}, at any time the rate
of increase in mean fitness due to natural selection is equal to the genetic variance in fitness. In a situation of co-evolution, however, as we are studying here, Fisher's theorem does not apply, and indeed, the mean fitness of each species may be increasing at some times and decreasing at others, despite the genetic variance in fitness always being nonnegative.

\item 
Theorem~\ref{thm:strongfail} raises the question of whether there are natural conditions (apart from those of~\cite{PilSch}) under which diversity is maintained in the strongest sense, i.e., Property~\ref{propPapaBear} holds. A natural possibility to consider is two-team zero-sum games in which there is no \emph{duality gap} as defined in~\cite{SchulVaz}. The converse is known to be false, namely, there are known to be games in which there is a duality gap and Property~\ref{propPapaBear} holds (see discussion prior to Theorem 2 in~\cite{PilSch}).
         
\item We should clarify the relationship of our results to those of Mehta, Panageas and Piliouras~\cite{MPP}.
In their setting, with time-invariant fitnesses, there must exist a pure Nash equilibrium, and they
prove that, under mild conditions, the system always converges to such an equilibrium.  In our setting, the
time-varying fitnesses typically eliminate pure Nash equilibria, and we prove that, in the absence
of such equilibria, the dynamics maintains diversity (in the sense of Theorem~\ref{thm:main}).  
Of course, if our Red Queen dynamics does have a pure Nash equilibrium, then it may converge to it
(as we demonstrate in Theorem~\ref{thm:main}(iii)).  Thus it is in cases where it is more beneficial to play
mixed strategies that diversity is useful for species-level survival.
\item Keeping in mind the connection between Red Queen dynamics and the MWU algorithm~\cite{CLPV}, our results show some of the dramatic dynamics that can occur when multiple agents are simultaneously implementing MWU.
\end{enumerate}

\end{document}